\documentclass[lettersize,journal]{IEEEtran}
\ifCLASSINFOpdf
\else
\fi
\usepackage{mathtools}
\usepackage{blindtext}
\usepackage{enumitem}
\usepackage{bbm, dsfont}
\usepackage{color}
\usepackage{xcolor}
\usepackage{tcolorbox}
\usepackage{balance}
\usepackage{amsthm}
\usepackage{nicematrix,tikz}
\usepackage{bm}
\usepackage[utf8]{inputenc}
\usepackage{booktabs,caption}
\usepackage[flushleft]{threeparttable}
\usepackage[english]{babel}
\usepackage{amssymb}
\usepackage[utf8]{inputenc} 
\usepackage[T1]{fontenc}
\usepackage{dblfloatfix}
\usepackage{url}
\usepackage{ifthen}
\usepackage{cite}
\usepackage{graphicx}
\usepackage{psfrag}
\usepackage{blkarray}
\usepackage{amsmath}

\def\vv{{\mathsf d}}
\def\vv{{\mathsf v}}
\def\nn{\nonumber}

\usepackage{accents}

\allowdisplaybreaks
\begin{document}
%
\title{
Capacity of the Binary Energy Harvesting Channel
 \thanks{
 The material in this paper was presented in part at the
 56th Annual Allerton Conference on Communication etc., Monticello, IL, USA, October 2018, and at the IEEE International Symposium on Information Theory, Los Angeles, CA, USA, June 2020.}
} 
\author{
    \IEEEauthorblockN{Eli Shemuel\IEEEauthorrefmark{1}, Oron Sabag\IEEEauthorrefmark{2}, Haim Permuter\IEEEauthorrefmark{1}}\\
    \IEEEauthorblockA{\IEEEauthorrefmark{1}Ben-Gurion University} 
    \IEEEauthorblockA{\IEEEauthorrefmark{2}The Hebrew University of Jerusalem}
}
\maketitle


%
\IEEEpeerreviewmaketitle

\newtheorem{question}{Question}
\newtheorem{claim}{Claim}
\newtheorem{guess}{Conjecture}
\newtheorem{definition}{Definition}
\newtheorem{fact}{Fact}
\newtheorem{assumption}{Assumption}
\newtheorem{theorem}{Theorem}
\newtheorem{lemma}{Lemma}
\newtheorem{ctheorem}{Corrected Theorem}
\newtheorem{corollary}{Corollary}
\newtheorem{proposition}{Proposition}
\newtheorem{remark}{Remark}
\newtheorem{example}{Example}

\def\cS{{\mathcal S}}
\def\cX{{\mathcal X}}
\def\cU{{\mathcal U}}
\def\cQ{{\mathcal Q}}
\def\cY{{\mathcal Y}}

\begin{abstract}
The capacity of a channel with an energy-harvesting (EH) encoder and a finite battery remains an open problem, even in the noiseless case. 
A key instance of this scenario is the binary EH channel (BEHC), where the encoder has a unit-sized battery and binary inputs. 
Existing capacity expressions for the BEHC are not computable, motivating this work, which determines the capacity to any desired precision via convex optimization. By modeling the system as a finite-state channel with state information known causally at the encoder, we derive single-letter lower and upper bounds using auxiliary directed graphs, termed $Q$-graphs. These $Q$-graphs exhibit a special structure with a finite number of nodes, $N$, enabling the formulation of the bounds as convex optimization problems. As $N$ increases, the bounds tighten and converge to the capacity with a vanishing gap of $O(N)$. For any EH probability parameter $\eta\in \{0.1,0.2, \dots, 0.9\}$, we compute the capacity with a precision of ${1e-6}$, outperforming the best-known bounds in the literature. Finally, we extend this framework to noisy EH channels with feedback, and present numerical achievable rates for the binary symmetric channel using a Markov decision process.
\end{abstract}
\begin{IEEEkeywords}
channel capacity, channels with feedback, convex optimization, energy harvesting, finite-state channel, $Q$-graphs.
\end{IEEEkeywords}

\section{Introduction}
Energy-harvesting (EH) communication models \cite{OptimalEnergyManagement2010,ozel2012optimal,UlkusYang12_energy_harvesting_comm12,Uluku12_Energy_harvesting_BC,AWGNchannelUnderTimeVarying2011,AchievingAWGNCapacity2012,tutuncuoglu2012optimum,mao2017CapAn,jog2014energy,Dong2015NearOptimal,Shaviv2016CapEHCfiniteBattery,tutuncuoglu2013binary,tutuncuoglu2014improved,tutuncuoglu2017binary} have gained significant attention due to their critical role in enabling sustainable and autonomous wireless communication systems. These models find applications in a wide range of emerging technologies, where transmitters harvest energy from the environment either for immediate use or to be stored in a battery for future transmissions, such as low-power wireless sensor networks \cite{visser2013rf,ma2017experimental} and modern unmanned aerial vehicle communications \cite{10452297_2024,EnergyEfficientUAV2020}.

Typically, EH models describe communication channels where the encoder is powered by a battery that is charged by an exogenous energy source according to an energy arrival process, as depicted in Fig.~\ref{fig:energyharvesting} (with or without output feedback). The channel inputs are constrained by the remaining energy stored in the battery. These unprecedented constraints differ fundamentally from conventional channels with average or peak power constraints, making the capacity analysis of EH models particularly difficult. The capacity has only been determined 
in extreme cases, such as when no battery is present \cite{AWGNchannelUnderTimeVarying2011} or when the battery is infinitely large \cite{AchievingAWGNCapacity2012}. However, for the intermediate scenarios of any finite-sized, non-zero battery, introduced in \cite{tutuncuoglu2012optimum} and further studied in \cite{mao2017CapAn,jog2014energy,Dong2015NearOptimal,Shaviv2016CapEHCfiniteBattery, tutuncuoglu2013binary,tutuncuoglu2014improved,tutuncuoglu2017binary}, the capacity remains an open problem.
\begin{figure}[t]
 \begin{center}
 \begin{psfrags}
     \psfragscanon
     \psfrag{M}[][][1]{$M$}
     \psfrag{N}[][][1]{$\hat{M}$}
     \psfrag{E}[][][1]{$E_i$}
     \psfrag{B}[][][1]{Battery\hspace{-0.7cm}}
     \psfrag{T}[][][0.98]{Encoder}
     \psfrag{C}[][][1]{Channel}
     \psfrag{R}[][][1]{Decoder}
     \psfrag{X}[][][1]{$X_i$}
     \psfrag{Y}[][][1]{$Y_i$}
      \psfrag{F}[][][1]{$Y_{i-1}$\hspace{-0.7cm}}
 \includegraphics[scale=0.38]{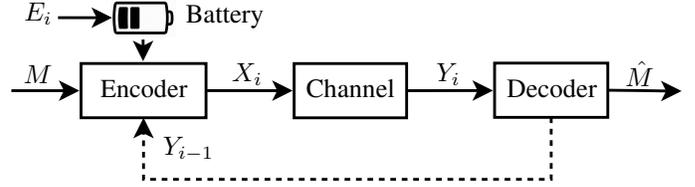}
 \caption{The energy-harvesting model with a finite battery.} \label{fig:energyharvesting}
 \psfragscanoff
 \end{psfrags}
 \end{center}
 \end{figure}

The capacity of a finite battery EH model with an AWGN channel is approximated within a constant gap in \cite{Dong2015NearOptimal,Shaviv2016CapEHCfiniteBattery}.
In \cite{jog2014energy}, this scenario is studied with a deterministic energy arrival process, and computable bounds on the capacity are obtained by calculating
the volume of feasible input vectors. Multi-letter capacity expressions are provided in \cite{mao2017CapAn} for a discrete memoryless channel (DMC) using the Verdú-Han framework \cite{Verdu94}, and in \cite{Shaviv2016CapEHCfiniteBattery} for any memoryless channel. However, these capacity expressions are difficult to evaluate due to the optimization over an infinite number of Shannon strategies \cite{shannon1958channels}. 
Recognizing the complexity of the finite-battery EH model's capacity problem, \cite{tutuncuoglu2013binary} introduces a simplified binary EH model, inspired by \cite{Popovski}, called the \textit{binary EH channel} (BEHC), to make progress toward deriving a computable capacity expression.

The BEHC, as its name suggests, operates with binary alphabets. After each channel use, energy is harvested in binary amounts -- either $0$ or $1$ unit. The battery is binary, meaning it is unit-sized and can be either empty or full, and the channel inputs are binary as well. Transmitting a $1$ requires one energy unit per channel use, whereas transmitting a $0$ consumes no energy. Thus, although the encoder can transmit a $0$ at any time, transmitting a $1$ is only possible when the battery is charged. Viewing the battery state as the channel state naturally defines the BEHC as a channel with state information (SI) that is available causally to the encoder but not to the decoder. However, the capacity cannot be achieved using the Shannon strategy \cite{shannon1958channels}, as the state process is not independent and identically distributed (i.i.d.) over time. Instead, the state evolution is affected by the full history of the channel inputs,
which may exhibit infinite memory. The potential for infinitely large memory in the SI, which is unknown to the decoder, makes the capacity problem challenging even in the noiseless case. This complexity arises because the uncertainty that causes communication errors stems from the state evolution, which incorporates the random EH process and imposes intricate constraints on the inputs. For instance, when the decoder receives a noiseless output 
$0$, it cannot determine whether the encoder intended to transmit a $0$ or was constrained to do so due to an empty battery.
 
The problem of the BEHC with an i.i.d. Bernoulli($\eta$) energy arrival process and noiseless channel is studied in \cite{tutuncuoglu2013binary,tutuncuoglu2014improved,tutuncuoglu2017binary}.
Throughout this paper, and as in these works, the BEHC refers to the noiseless channel scenario unless otherwise stated. They demonstrate that the BEHC is equivalent to a timing channel, analogous to the telephone signaling channel in \cite{anantharam1996bits,EntropyTimingCap,ITcap_queues}, by modeling the time differences between consecutive $1$s transmitted through the channel. In the equivalent representation, a single-letter expression for the capacity of the BEHC is derived. 
However, evaluating this expression remains difficult because it involves an auxiliary random variable (RV) with infinite cardinality, and its capacity-achieving distribution is unknown. Based on specific choices of distributions, analogous to lattice coding for the timing channel, \cite{tutuncuoglu2017binary} provides numerical achievable rates. Additionally, two upper bounds are proposed: one using a genie-aided method and the other quantifying the leakage of SI to the decoder \cite{TavanBitsThrBuf}. Numerical results are evaluated for various EH probability parameters $\eta\in \{0.1,0.2, \dots, 0.9\}$, showing a small gap between the lower and upper bounds. This gap is larger for small $\eta$ values and decreases as $\eta$ increases.

In this work, we also consider the BEHC and provide a solution for practically computing the capacity with any desired precision using convex optimization. First, recall that the remaining energy in the battery serves as the channel state, which is causally known at the encoder and has memory. Consequently, we view the BEHC as a finite-state channel (FSC), where the memory is encapsulated in the channel state. The BEHC is a special case of FSCs with feedback and SI causally available at the encoder, as introduced in \cite{shemuel2024finite} and depicted in Fig.~\ref{fig:setting}.
Since the channel is noiseless, feedback does not increase the capacity, allowing us to utilize known results from \cite{shemuel2024finite}. They derive the feedback capacity of the general setting as the directed information between auxiliary RVs with
memory to the channel outputs, and provide computable lower bounds via the $Q$-graph method \cite{Sabag_UB_IT} and a Markov decision process (MDP) formulation. 

Here, we construct a single-letter lower bound for the BEHC based on a specific $Q$-graph with $N$ nodes, which has a special structure given $N$. Although a $Q$-graph upper bound is not derived for the general setting in \cite{shemuel2024finite} due to the lack of cardinality bound on the auxiliary RV, we successfully derive a computable single-letter $Q$-graph upper bound for the BEHC. Similar to the lower bound, the upper bound is established on a specific $N$-node $Q$-graph with a special structure. Furthermore, we formulate both the $Q$-graph lower and upper bounds as convex optimization problems, making them computable via convex optimization algorithms. We show that the sequences of these bounds coincide and converge to the capacity of the BEHC as $N$ approaches infinity, with the gap between them scales as $O(N)$. This enables us to compute the capacity with any desired precision. For any EH parameter $\eta\in \{0.1,0.2, \dots, 0.9\}$, we compute the capacity with a precision of ${1e-6}$ and compare it to the tightest numerical bounds known in the literature from \cite{tutuncuoglu2017binary}. Our results outperform the previous bounds across all $\eta$ values, with particularly significant improvements for smaller $\eta$. 

\begin{figure}[t]
\begin{center}
\begin{psfrags}
    \psfragscanon
    \psfrag{E}[][][1]{$M$}
    \psfrag{S}[][][1]{\begin{tabular}{@{}l@{}}
    $\; S^{i-1}$
    \end{tabular}}
    \psfrag{A}[\hspace{2cm}][][1]{Encoder}
	 \psfrag{F}[\hspace{1cm}][][1]{$X_i$}
	 \psfrag{B}[\hspace{2cm}]{{$P(s_i,y_i|x_i,s_{i-1})$}}
	 \psfrag{G}[][][1]{$Y_i$}
	 \psfrag{C}[\hspace{2cm}][][1]{Decoder}
	 \psfrag{K}[][][1]{$\hat{M}$}
	 \psfrag{H}[\hspace{2cm}][][1]{$Y_i$}
	 \psfrag{D}[\hspace{2cm}][][0.95]{Unit-Delay}
	 \psfrag{J}[\vspace{2cm}\hspace{2cm}][][1]{$Y_{i-1}$}
	 \psfrag{L}[\hspace{2cm}][][1]{Finite-State Channel}
	 \psfrag{I}[][][1]{}
\includegraphics[scale=0.63]
{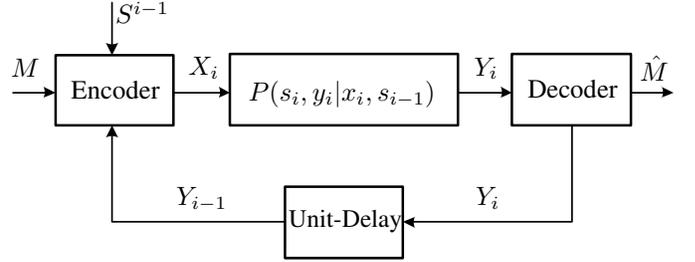}
\caption{
FSC with feedback and SI available causally to the encoder. At time $i$, the current state~$S_{i-1}$ influences output~$Y_i$.}
\label{fig:setting}
\psfragscanoff
\end{psfrags}
\end{center}
\end{figure}

Additionally, we extend the BEHC framework to noisy channels with feedback and demonstrate that the MDP formulation in \cite{shemuel2024finite} can be applied to any DMC with feedback. We apply the value iteration algorithm (VIA) to numerically evaluate achievable rates for the BEHC over a binary symmetric channel (BSC). To the best of our knowledge, no computable bounds on the feedback capacity have previously been derived in the literature for an EH model with a noisy channel.

The rest of the paper is structured as follows:
Section~\ref{sec:Com_Setup} defines the communication setup. Section~\ref{sec:MainResults} presents the main results. Section~\ref{sec:EHconvex} describes the convex optimization approach used to calculate the lower and upper bounds on the capacity. Section~\ref{sec:proof_main_theorems} provides proofs of the main results. Section~\ref{sec:Noisy} presents numerical achievable rates for the BEHC with a BSC. Finally, Section~\ref{sec:conclusions} concludes the paper.
\section{Preliminaries and the Channel Model}
\label{sec:Com_Setup}

\subsection{Notation and Definitions}
Lowercase letters denote sample values (e.g., $x,y$), and uppercase letters denote discrete RVs (e.g., $X,Y$). Subscripts and superscripts denote vectors as follows: $x_i^j=(x_i,x_{i+1},\dots,x_j)$ and $X_i^j=(X_i,X_{i+1},\dots,X_j)$ for $1\leq i \leq j$, whereas $x^n\triangleq x_1^n$ and $X^n\triangleq X_1^n$. The notation $\mathbf{0}^N$ denotes a vector of length \( N \) with all entries equal to $0$.
Calligraphic symbols (e.g., $\cX,\cY$) represent alphabets, with $|\cX|$ indicating the alphabet's cardinality. For two RVs $X,Y$, the probability mass function (PMF) of $X$ is expressed as $P(X=x)$, while the conditional PMF given $Y=y$ is written as $P(X=x|Y=y)$, and their joint PMF is denoted by $P(X=x,Y=y)$. For these PMFs, the shorthands $P(x), P(x|y)$ and $P(x,y)$ are used. The indicator function $\mathbbm{1}(b)$, given a binary condition $b$, equals $1$ if $b$ is true and $0$ otherwise.
Given two integers $n$ and $m$ such that $n \le m$, the discrete interval $[n:m]$ is defined as $\triangleq \{n,n+1,\dots,m\}$. 
For any $\alpha \in [0,1]$, we define $\bar{\alpha}=1-\alpha$.

Logarithms are in base $2$; hence, entropy is measured in bits. The binary entropy function is defined by $H_2(p)\triangleq -p \log p -\bar{p} \log \bar{p}, \; p\in [0,1]$.
The \textit{directed information} between $X^N$ and $Y^N$
is defined as $I(X^n\rightarrow Y^n)\triangleq \sum_{i=1}^{n} I(X^i;Y_i|Y^{i-1})$.
The \textit{causal conditioning distribution} 
is defined as $P(x^n||y^{n-1})\triangleq \prod_{i=1}^n P(x_i|x^{i-1},y^{i-1})$. 

A \textit{FSC} is defined by finite alphabets for its input ($\cX$), output ($\cY$), and state ($\cS$) variables. At each time $i$, the corresponding symbols are $X_i,Y_i$ and $S_{i-1}$, and the channel, being time-invariant, is governed by $P_{S^+,Y|X,S}$, satisfying 
\begin{equation}
P(s_i,y_i|x^i,s_0^{i-1},y^{i-1}) = P_{S^+,Y|X,S}(s_i,y_i|x_i,s_{i-1}).
\label{eq:BasicFSCMarkov}
\end{equation}

\begin{definition}[Connectivity]\label{def:connectivity} \cite[Def.~2]{Permuter06_trapdoor_submit}
A FSC is called \textit{strongly connected} if for all $s',s \in \cS$ there exists an integer $T(s)$ and input distribution of the form $\{P(x_i|s_{i-1})\}_{i=1}^{T(s)}$ that may depend on $s$, such that $\sum_{i=1}^{T(s)} P(S_i=s|S_0=s')>0$.
\end{definition}

\vspace{-0.3cm}
\subsection{Preliminary on Quantized-graphs}
\label{subsection:prelQgraph}
The $Q$-graph is a method for deriving single-letter lower and upper bounds on the capacity of a setting with feedback from a multi-letter expression. This technique is established on \textit{Quantized-graphs} ($Q$-graphs), which are directed, connected graphs with a finite number of nodes $|\cQ|$. Each node represents a unique value $q \in \cQ$ and is associated with exactly $|\cY|$ outgoing edges, each labeled with a distinct symbol from $\cY$. By definition, starting from an initial node, $q_0$, and given an output sequence, $y^i$, traversing the corresponding labeled edges uniquely determines a final node, $q_i$. This transition process is governed by a time-invariant function, $g:\cQ\times\cY\to\cQ$, where the next node is determined by the current node and the channel output.
Recursively, the final node after $i$ steps is given by the function composition
$g^i(q_0,y^i) = g(g^{i-1}(q_0,y^{i-1}),y_i)$. A visual instance of a $Q$-graph is illustrated in Fig.~\ref{fig:q_graph_ex}.
\begin{figure}[t]
\centering
    \psfrag{Q}[][][1]{$Y=0$}
    \psfrag{E}[][][1]{$Y=1$}
    \psfrag{F}[][][1]{$Y=2$}
    \psfrag{O}[][][1]{$Y=0/1/2$}
    \psfrag{L}[][][1]{$Q=1$}
    \psfrag{H}[][][1]{$Q=0$}
    \includegraphics[scale = 0.5]{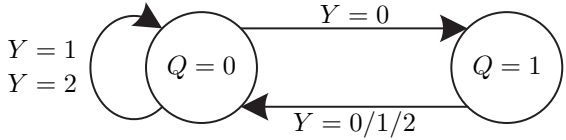}
    \caption{An instance of a $2$-node $Q$-graph with ternary output.}
    \label{fig:q_graph_ex}
\end{figure}

\vspace{-0.15cm}
\subsection{The Energy-Harvesting Communication Model}
We consider the EH model depicted in Fig.~\ref{fig:energyharvesting} with a unit-sized battery, i.e., it can store at most $B_{max}=1$ energy unit. At time $i$,
the current battery state is represented by
$S_{i-1}\in \{0,1\}$, which is causally known to the encoder but remains unknown to the decoder, and $X_i\in\{0,1\}$ denotes the transmitted input symbol during channel use $i$. The encoder is always capable of transmitting $X_i=0$ regardless of the battery state. However, sending $X_i=1$ is only feasible when the battery is charged ($S_{i-1}=1$). Thus, if the battery is empty ($S_{i-1}=0$), the input is constrained to $X_i=0$. After each transmission, the encoder harvests an energy unit according to an i.i.d. arrival process $E_i\sim$\text{Bern}($\eta$), where $\eta\in [0,1]$ denotes the success probability. 
If $S_{i-1}=0$, a successfully harvested energy unit charges the battery to $S_i=1$, enabling a future transmission of $X=1$. However, if $S_{i-1}=1$, the harvested energy is lost as it cannot be used without being stored.\footnote{This scenario is referred to as the \textit{transmit first model} in \cite{tutuncuoglu2017binary}, as the encoder first transmits $X_i$ and then harvests energy $E_i$, such that the harvested energy is not immediately available for transmission. A \textit{harvest first} model, in which the order is reversed (harvest then transmit), is considered in \cite{mao2017CapAn}.}
For simplicity, and without loss of generality (see \cite[Prop.~1]{Shaviv2016CapEHCfiniteBattery}), we assume that the battery starts in an empty state, i.e., $S_0=0$.
The evolution of the battery state is governed by the equation
\begin{equation}
S_{i}=\min\{S_{i-1}-X_i + E_i,1\}, \label{eq:BEHC_state_evolution}
\end{equation}
where $X_i = 0$ if $S_i = 0$ (or equivalently, $X_i\le S_{i-1}$).

Our primary focus is on the noiseless channel case, where $Y_i=X_i$ holds for all $i$. This case is referred to as the \textit{noiseless BEHC}, or simply \textit{BEHC}. Additionally, we consider a generalized scenario with noise: a DMC $P(y|x)$ in the presence of feedback, referred to as a \textit{noisy BEHC}.
At time $i$, the encoder observes not only the message $M$ and the current battery state $S_{i-1}$, but also the output feedback $Y_{i-1}$. 
The definitions of 
a \textit{$(2^{nR},n)$ code}, the \textit{average probability of error}, \textit{achievable rate} and \textit{capacity} are standard (see \cite{shemuel2024finite}). The capacity of the noiseless BEHC is denoted by $C_{\text{BEHC}}$ throughout this paper.

Determining $C_{\text{BEHC}}$ in a computable form for any harvesting parameter $\eta$ remains an open problem, and resolving it is the primary objective of this work. The challenge arises from the intricate constraints at the encoder due to \eqref{eq:BEHC_state_evolution}, while the decoder remains unaware of the battery state, even when receiving noiseless outputs. Since feedback does not increase the capacity of any noiseless channel model, $C_{\text{BEHC}}$ corresponds to a special case of noisy BEHCs with feedback. A noisy BEHC with feedback can be viewed as a strongly connected FSC $P_{S^+,Y|X,S}$ with feedback and SI available causally at the encoder, as illustrated in Fig.~\ref{fig:setting}, and specified~by
\begin{align}
&P_{S^+,Y|X,S}(s_i,y_i|x_i,s_{i-1}) =P(y_i|x_i)P(s_i|x_i,s_{i-1}) \nonumber\\
&=P(y_i|x_i)\sum\nolimits_{j\in\{0,1\}} P(E_i=j)P(s_i|x_i,s_{i-1},E_i=j) \nonumber\\
&=P(y_i|x_i)\big(
\bar{\eta} \mathbbm{1}\{s_i=s_{i-1}-x_i\} + \eta \mathbbm{1}\{s_i=1\}\big), \label{eq:stateEvolutionEH}
\end{align}
where $x_i\le s_{i-1}$, and in the noiseless channel scenario, the FSC $P_{S^+|X,S}$ is characterized by
\begin{align}
P_{S^+|X,S}(s_i=0|x_i,s_{i-1}) & =\bar{\eta} \mathbbm{1}\{x_i=s_{i-1}\}.
\label{eq:stateEvolutionNoiselessEH}
\end{align}
The capacity of such general FSC $P_{S^+,Y,X,S}$ is denoted by $C_{\text{fb-csi}}$, and is studied in \cite{shemuel2024finite}.

\vspace{-0.28cm}
\subsection{A Lower Bound on the Feedback Capacity of FSCs With SI Available Causally at the Encoder}
In \cite{shemuel2024finite}, a single-letter lower bound on $C_{\text{fb-csi}}$ is derived using $Q$-graphs, forming the foundation of our main results.
Given a $Q$-graph, an auxiliary RV $U$ with a specified cardinality $|\cU|$, a strategy function $f: \mathcal U \times \mathcal S \to \mathcal X$ and an input distribution $P(u^+|u,q)$, the transition matrix $P(s^+,u^+,q^+|s,u,q)$ is defined~as
\begin{align}
\label{eq:suq_transition}
& P(s^+,u^+,q^+|s,u,q) = \sum\nolimits_{x,y} P(u^+|u,q) \mathbbm{1}\{x=f(u^+,s)\} \nn\\
&\quad \times \mathbbm{1}\{ q^+=g(q,y) \}  P_{S^+,Y|X,S}(s^+,y|x,s).
\end{align}
The set $\mathcal{P}_{\pi}$ denotes all $P(u^+|u,q)$ that induce a transition matrix \eqref{eq:suq_transition} with a unique stationary distribution over $(S,U,Q)$, denoted by $\pi(s,u,q)$. 
An input distribution $P(u^+|u,q)\in \mathcal{P}_\pi$ is said to be \textit{BCJR-invariant} if it satisfies the Markov chain
\begin{align}
 (U^+,S^+)-Q^+-(Q,Y). \label{eq:BCJR_Def}
\end{align}
\begin{theorem}(\!\!\cite[Th.~5]{shemuel2024finite}):
 \label{theorem:qgraph_LB}
For any $Q$-graph, given a fixed finite cardinality $\left|\mathcal{U}\right|$ ($U^+,U \in \cU$) and a function $f: \mathcal U \times \mathcal S \to \mathcal X$, the feedback capacity is lower bounded by
\begin{align}\label{eq:Theorem_Lower}
C_{\text{fb-csi}}\geq I(U^+,U;Y|Q),
\end{align}
for all $P(u^+|u,q)\in\mathcal{P}_\pi$ that are BCJR-invariant \eqref{eq:BCJR_Def}, where the joint distribution is given by
\begin{align}
\label{eq:Qgraph_joint_dist}
  P&_{S,U,Q,X,Y,S^+,U^+,Q^+} = \pi_{S,U,Q} P_{U^+|U,Q} \mathbbm{1}\{X = f(U^+,S)\} \nn\\ &\quad \times P_{S^+,Y|X,S} \mathbbm{1}\{Q^+ = g(Q,Y)\}. 
\end{align}
\end{theorem}

For the policy's auxiliary RVs $U,U^+$, we use the notations $U^{(Q=q)},U^{+(U=u,Q=q)}$ throughout the paper to denote them conditioned on $Q=q$ or $(U=u,Q=q)$, respectively, i.e., $U \mid Q=q$ and $U^+ \mid (U=u,Q=q)$.

\section{Main Results}
\label{sec:MainResults}
In this section, we present our main results. 
Using the $Q$-graph method, we derive sequences of single-letter convex lower and upper bounds on $C_{\text{BEHC}}$, which converge to the capacity as their limits coincide. Each term in these sequences is defined using the auxiliary RVs $(U,U^+,Q,Q^+)$, whose cardinality is $N+1$. To emphasize their dependence on $N$, we denote them as $(U^{(N)},U^{+(N)},Q^{(N)},Q^{+(N)})$, though we often omit the superscript $^{(N)}$ for simplicity.

We begin with the following theorem, which establishes the lower bound sequence based on Theorem~\ref{theorem:qgraph_LB}.
\begin{figure}[t]
\resizebox{\columnwidth}{!}{
\begin{psfrags}
    \psfragscanon
    \psfrag{E}[][][0.8]{$M$}
    \psfrag{A}[\hspace{2cm}][][1]{$Q=0$}
    \psfrag{B}[\hspace{2cm}][][1]{$Q=1$}
    \psfrag{C}[\hspace{2cm}][][1]{$Q=2$}
    \psfrag{G}[\hspace{2cm}][][0.75]{}
    \psfrag{D}[\hspace{2cm}][][0.65]{$Q=N-1$}
    \psfrag{E}[\hspace{2cm}][][1]{$Q=N$}
    \psfrag{F}[\hspace{2cm}][][0.75]{$X=0$}
    \psfrag{M}[\hspace{2cm}][][1]{$...$}
    \psfrag{H}[][][0.75]{$X=1$\hspace{0.1cm}}
    \psfrag{I}[\hspace{2cm}][][1]{$...$\hspace{0cm}}
    \psfrag{J}[\hspace{2cm}][][0.75]{$X=1$\hspace{-0.5cm}}
    \psfrag{M}[\hspace{2cm}][][0.75]{$X=1$\hspace{-0.30cm}}
    \psfrag{K}[\hspace{2cm}][][0.75]{$X=1$\hspace{0cm}}
    \psfrag{N}[\hspace{2cm}][][0.75]{$X=0$\hspace{0.5cm}}
    \psfrag{L}[\hspace{2cm}][][0.75]{$X=0,1$\hspace{0cm}}
    \includegraphics[scale=0.65]{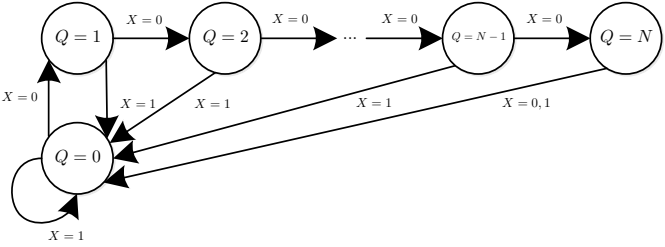}
\end{psfrags}
}
\caption{\small A $Q$-graph used for lower bounding $C_{\text{BEHC}}$ in Theorem~\ref{thr:EH_cvx_lower}.}
\label{fig:EH_LB_Graph}
\end{figure}

\begin{theorem}[Convex Lower Bound]
\label{thr:EH_cvx_lower}
For any integer $N\ge0$, the capacity of the BEHC is lower bounded by
\begin{align}
C_{\text{BEHC}}\geq I(U^{+(N)},U^{(N)};X\mid Q^{(N)}), \label{eq:EH_LB_cvx}
\end{align}
where $Q^{(N)}$ is defined on the vertices of the $Q$-graph illustrated in Fig.~\ref{fig:EH_LB_Graph}, with $|\cQ^{(N)}|= N+1$, the auxiliary RV sets are given by
\begin{align}           
 &\mspace{-10mu}\cU^{(Q=q)}\mspace{-5mu}=\mspace{-5mu}[0\mspace{-2mu}:\mspace{-2mu}q], \cU^{+(U=u,Q=q)}\mspace{-5mu}=\mspace{-5mu}\{0,u\mspace{-5mu}+\mspace{-5mu}1\}, \forall q\mspace{-2mu}\in\mspace{-2mu}[0\mspace{-2mu}:\mspace{-2mu}N\mspace{-3mu}-\mspace{-3mu}1], \mspace{-12mu}\label{eq:Uset_initialNodes}\\
 &\mspace{-10mu}\cU^{(Q=N)}=[0:N],  \cU^{+(Q=N)}=\{0\}, \label{eq:Uset_LastNode}    
\end{align}
and the joint distribution $P_{S,U,Q,X,S^+,U^+,Q^+}$ is given by \eqref{eq:Qgraph_joint_dist}, where $Y=X$. The FSC characterization $P_{S^+|X,S}$ is given in \eqref{eq:stateEvolutionNoiselessEH},
and the function $f: \mathcal U \times \mathcal S \to \mathcal X$ is defined as
\begin{align}
     f(U^+,S)&=S \mathbbm{1}\{U^+=0\}.
          \label{eq:x_behc}
\end{align}
The induced marginal $\pi_{S|U,Q}$ remains constant for any choice of policy $P_{U^+|U,Q}$, given \eqref{eq:Uset_initialNodes}–\eqref{eq:Uset_LastNode}, and is given by
\begin{align}
\label{eq:s_given_uq_lb}
\pi(S=0|u,q)=\bar{\eta}^{u+1}, \quad \forall q\in \cQ, \forall u\in \mathcal{U}^{(Q = q)}.
\end{align}
Optimizing over $P_{S,U,Q,X,S^+,U^+,Q^+}$ can be formulated as a convex optimization problem, detailed in Eq.~\eqref{optProb_EH_LB}, Section~\ref{subsec:EHconvex_lower}. Moreover, the optimized lower bound in \eqref{eq:EH_LB_cvx} converges to $C_{\text{BEHC}}$ as $N\to \infty$.
\end{theorem}

The proof of Theorem \ref{thr:EH_cvx_lower} is provided in Section~\ref{sec:QgraphLB}. Note that $|\cU^{(Q=q)}|$ increases by $1$ with $q$. In general, $|\cU^{(N)}|=|Q^{(N)}|=N+1$. Furthermore, \eqref{eq:Uset_initialNodes}–\eqref{eq:Uset_LastNode} imply that the policy $P_{U^+|U,Q}$ is subject to
\begin{align}
P(u^+ \mid u, q) &= 0, \quad \forall q \in [0:N-1], \; \forall u \in \mathcal{U}^{(Q = q)}, \nonumber \\
&\quad \forall u^+ \in \mathcal{U}^{+(Q = q)} \setminus \{0, u+1\}, \label{eq:policy_first_nodes} \\
P_{U^+ \mid U, Q}(0 \mid u, N) &= 1, \quad \forall u \in \mathcal{U}^{(Q = N)}. \label{eq:policy_last_node}
\end{align}
In particular, this policy structure, combined with \eqref{eq:x_behc}, always attempts to transmit $X=1$ at node $Q=N$, with success determined by whether the battery is charged. While the BCJR constraint of the policy in its general form \eqref{eq:BCJR_Def} translates into non-linear equality constraints, Theorem~\ref{thr:EH_cvx_lower} demonstrates that the $Q$-graph lower bound on $C_{\text{BEHC}}$ in \eqref{eq:EH_LB_cvx} can be formulated as a convex optimization problem. 

We now turn to the upper bound sequence. Although a general single-letter $Q$-graph upper bound on $C_{\text{fb-csi}}$ was not derived in \cite{shemuel2024finite} due to the absence of a general cardinality bound on $\cU$, the following theorem establishes a sequence of such bounds specifically for $C_{\text{BEHC}}$, with a finite $\cU$ and a $Q$-graph modified from Fig.~\ref{fig:EH_LB_Graph} by a single edge.

\begin{figure}[t]
\resizebox{\columnwidth}{!}{
\begin{psfrags}
    \psfragscanon
    \psfrag{E}[][][0.8]{$M$}
    \psfrag{A}[\hspace{2cm}][][1]{$Q=0$}
    \psfrag{B}[\hspace{2cm}][][1]{$Q=1$}
    \psfrag{C}[\hspace{2cm}][][1]{$Q=2$}
    \psfrag{G}[\hspace{2cm}][][0.75]{$X=0$}
    \psfrag{D}[\hspace{2cm}][][0.65]{$Q=N-1$}
    \psfrag{E}[\hspace{2cm}][][1]{$Q=N$}
    \psfrag{F}[\hspace{2cm}][][0.75]{$X=0$}
    \psfrag{M}[\hspace{2cm}][][1]{$...$}
    \psfrag{H}[][][0.75]{$X=1$\hspace{0.1cm}}
    \psfrag{I}[\hspace{2cm}][][1]{$...$\hspace{0cm}}
    \psfrag{J}[\hspace{2cm}][][0.75]{$X=1$\hspace{-0.5cm}}
    \psfrag{M}[\hspace{2cm}][][0.75]{$X=1$\hspace{-0.30cm}}
    \psfrag{K}[\hspace{2cm}][][0.75]{$X=1$\hspace{0cm}}
    \psfrag{N}[\hspace{2cm}][][0.75]{$X=0$\hspace{0.5cm}}
    \psfrag{L}[\hspace{2cm}][][0.75]{$X=1$\hspace{0cm}}
    \includegraphics[scale=0.65]{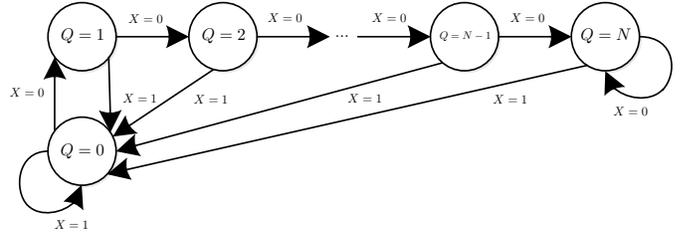}
\end{psfrags}
}
\caption{\small A $Q$-graph for upper bounding $C_{\text{BEHC}}$ in Theorem~\ref{thr:EH_cvx_upper}.}
\label{fig:EH_UB_Graph}
\end{figure}

\begin{theorem}[Convex Upper Bound]
\label{thr:EH_cvx_upper}
For any integer $N\ge0$, the capacity of the BEHC is upper bounded by 
\begin{align}
C_{\text{BEHC}}\leq \sup_{P(u^+|u,q)\in\mathcal{P}_\pi} I(U^{+(N)},U^{(N)};X \mid Q^{(N)}), \label{eq:EH_UB_cvx}
\end{align}
where $Q^{(N)}$ is defined on the vertices of the $Q$-graph illustrated in Fig.~\ref{fig:EH_UB_Graph}, with $|\cQ^{(N)}|= N+1$, the auxiliary RV sets are given by \eqref{eq:Uset_initialNodes} and
\begin{align}
    &\cU^{(Q=N)}=[0:N], \; \cU^{+(Q=N)}=\{0,1\}, \label{eq:Uset_LastNodeUB}
\end{align}
the joint distribution is given by
\begin{align}
\label{eq:Qgraph_joint_distUB}
  &P_{S,U,Q,X,S^+,U^+,Q^+}= \pi_{S,U,Q} P_{U^+|U,Q} \mathbbm{1}\{X = f(U^+,S)\}  \nn\\
  & \times P_{S^+|X,S,Q} \mathbbm{1}\{Q^+  = g(Q,X)\}, \\    
  &\mspace{-10mu}P_{S^+|X,S,Q}\mspace{-4mu} \triangleq \mspace{-4mu}
  \begin{cases}
     \mathbbm{1}\{S^+=1\}, & \text{if } Q\mspace{-4mu}\in\mspace{-4mu}\{N\mspace{-4mu}-\mspace{-4mu}1,\mspace{-3mu}N\}, \mspace{-4mu}X\mspace{-4mu}=\mspace{-4mu}0\\    
    P_{S^+|X,S} \text{ in \eqref{eq:stateEvolutionNoiselessEH}}
    & \text{otherwise}
   \end{cases}, \label{eq:modifiedChannelLaw}   
\end{align}
and the function $f(U^+,S)$ is defined in \eqref{eq:x_behc}.
The induced marginal $\pi_{S|U,Q}$ is constant for any choice of policy $P_{U^+|U,Q}$, given \eqref{eq:Uset_initialNodes} and \eqref{eq:Uset_LastNodeUB}, and is given by
\begin{align}    
\mspace{-15mu} \pi(S=0|u,q)&=
   \begin{cases}
      \bar{\eta}^{u+1},  &\mspace{-10mu} \forall q\in [0:N\mspace{-5mu}-\mspace{-5mu}1], \forall u\in \cU^{(Q=q)}\\
      0, &\mspace{-10mu}  q=N, \forall u\in \cU^{(Q=N)}
   \end{cases}. \label{eq:s_given_uq_upper}
\end{align}
Optimizing over the joint distribution $P_{S,U,Q,X,S^+,U^+,Q^+}$ can be formulated as a convex optimization problem, as detailed in Eq.~\eqref{optProb_EH_UB}, Section~\ref{subsec:EHconvex_upper}. Moreover, the upper bound in \eqref{eq:EH_UB_cvx} converges to $C_{\text{BEHC}}$ as $N\to \infty$.
\end{theorem}

The proof of Theorem~\ref{thr:EH_cvx_upper} is provided in Section~\ref{sec:QgraphUB}.
The key difference between Theorem~\ref{thr:EH_cvx_lower} and Theorem~\ref{thr:EH_cvx_upper} lies in the treatment of the last node $Q=N$. Specifically, in Fig.~\ref{fig:EH_LB_Graph}, the outgoing edge $'X=0'$ from $Q=N$ leads to $Q=0$, whereas in Fig.~\ref{fig:EH_UB_Graph}, it loops back to $Q=N$. Furthermore, \eqref{eq:Uset_LastNodeUB} implies that the policy $P_{U^+|U,Q}$ in Theorem~\ref{thr:EH_cvx_upper} is subject to 
\begin{align}
    P_{U^+|U,Q}(u^+|u,N)=0, \forall u\in \cU^{(Q=N)}, \forall u^+\in [2:N], \label{eq:policy_last_node_UB}
\end{align}
rather than \eqref{eq:policy_last_node} in the lower bound. Unlike the lower bound, the upper bound policy is neither obligated to attempt transmitting $X=1$ in the last node nor satisfies the BCJR property \eqref{eq:BCJR_Def} (particularly at $Q=N$, as can be verified). Furthermore, $P_{S^+|X,S,Q}$ in \eqref{eq:modifiedChannelLaw}, which replaces $P_{S^+|X,S}$ from Theorem~\ref{thr:EH_cvx_lower} in both the joint distribution and the transition matrix \eqref{eq:suq_transition}, differs at the last node $Q^+\mspace{-3mu}=\mspace{-3mu}N\mspace{-3mu}=\mspace{-3mu}g(Q\mspace{-3mu}=\mspace{-3mu}N-1,X\mspace{-3mu}=\mspace{-3mu}0)$, where $S^+\mspace{-3mu}=\mspace{-3mu}1$. This implies that when $Q\mspace{-3mu}=\mspace{-3mu}N$, reached after $N$ consecutive zero inputs, the battery is charged with probability~$1$, as also indicated in Eq.~\eqref{eq:s_given_uq_upper}. Thus, for any integer $N\ge0$, we implicitly consider a scenario whose capacity upper bounds $C_{\text{BEHC}}$, and aids in proving Theorem~\ref{thr:EH_cvx_upper}. 
Notably, an upper bound based on a Markov $Q$-graph of order $N$ could also be established. However, we focus on the $Q$-graph depicted in Fig.~\ref{fig:EH_UB_Graph}, since its number of nodes scales linearly with $N$ rather than exponentially.

The combination of Theorems~\ref{thr:EH_cvx_lower} and~\ref{thr:EH_cvx_upper} enables the evaluation of $C_{\text{BEHC}}$ using convex optimization algorithms, with a decreasing gap between the lower and upper bounds as $N$ increases. For any harvesting probability $\eta\in \{0.1,0.2,\dots,0.9\}$, we evaluate $C_{\text{BEHC}}(\eta)$ with a precision of ${1e-6}$ using CVX\footnote{\textit{CVX} is a MATLAB-based modeling system for convex optimization. We provide MATLAB code for evaluating the convex lower and upper bounds detailed in \eqref{optProb_EH_LB} and \eqref{optProb_EH_UB}, respectively, for any parameter $\eta$ and integer $N>0$. This code is openly accessible at the following \textit{GitHub} repository:\\ \url{https://github.com/Eli-BGU/Energy-Harvesting-CVX}.}. Our results are compared to the tightest bounds reported in the literature from \cite{tutuncuoglu2017binary}.
The results are summarized in Table~\ref{table:noiseless_BEHC_full}, where our lower and upper bounds are denoted by \textit{$Q$-graph LB} and \textit{$Q$-graph UB}, respectively. We also present the number of nodes required to achieve the specified precision, $|\cQ^{(N)}|=N+1$. Notably, smaller harvesting probabilities $\eta$ require a larger $N$ to achieve the same precision. This trend arises because, for larger $\eta$, when the encoder intends to transmit $N$ consecutive zero inputs, the battery is more likely to be charged afterward, and 
conversely, for smaller $\eta$, it is less likely. For all investigated values of $\eta$, our bounds outperform the best-known results in the literature. While the results from \cite{tutuncuoglu2017binary} are near-tight for larger values of $\eta$, our method provides accurate solutions across all regimes of $\eta$.

\begin{table}[t]
\caption{Bounds on the capacity of the BEHC with i.i.d. harvesting probability $E_i\sim\text{Bern}(\eta)$.}
\label{table:noiseless_BEHC_full}
\centering
 \begin{tabular}{|c||c||c|c|c||c|}
 \hline
 $\eta$ & \textit{LB} \cite{tutuncuoglu2017binary}& ${Q\text{-graph LB}}$ & $|\cQ^{(N)}|$ & ${Q\text{-graph UB}}$ & \textit{UB} \cite{tutuncuoglu2017binary}\\
 \hline
 \textbf{0.1} & 0.2317 & \textbf{0.234150} & 100 &\textbf{0.234151} & 0.2600\\
 \hline
 \textbf{0.2} & 0.3546 & \textbf{0.360193} & 70 & \textbf{0.360194} & 0.3871\\
 \hline
 \textbf{0.3} & 0.4487 & \textbf{0.457051} & 40 & \textbf{0.457052} & 0.4740\\
 \hline
 \textbf{0.4} & 0.5297 & \textbf{0.538820} & 30 & \textbf{0.538821} & 0.5485\\
 \hline
 \textbf{0.5} & 0.6033 & \textbf{0.610944} & 20 & \textbf{0.610945} & 0.6164\\
 \hline
 \textbf{0.6} & 0.6729 & \textbf{0.678468} & 16 & \textbf{0.678469} & 0.6807\\ 
 \hline
 \textbf{0.7} & 0.7403 & \textbf{0.743533} & 12 & \textbf{0.743534} & 0.7442\\ 
 \hline
 \textbf{0.8} & 0.8088 & \textbf{0.810034} & 8 & \textbf{0.810035} & 0.8101\\
 \hline
 \textbf{0.9} & 0.8845 & \textbf{0.884596} & 7 & \textbf{0.884597} & 0.8846\\
 \hline
 \end{tabular}
\end{table}

\section{Convex Optimization Formulations}
\label{sec:EHconvex}
In this section, we formulate the convex optimization problems for the lower and upper bounds on $C_{\text{BEHC}}$ provided in Theorem~\ref{thr:EH_cvx_lower} and Theorem~\ref{thr:EH_cvx_upper}, and we prove their convexity.

\subsection{Convex Lower Bound on $C_{\text{BEHC}}$}
\label{subsec:EHconvex_lower}
Prior to outlining the convex optimization problem corresponding to the lower bound sequence in Theorem~\ref{thr:EH_cvx_lower}, we first describe its $Q$-graph with $|\cQ|=N+1, N>0$, as illustrated in Fig.~\ref{fig:EH_LB_Graph}. For all nodes $Q\in\cQ=[0:N]$, the outgoing edge $'X=1'$ leads to node $Q=0$. For any node $Q\in [0: N-1]$, the outgoing edge $'X=0'$ leads to node $Q^+=Q+1$. However, for the last node $Q=N$, the outgoing edge $'X=0'$ leads back to the initial node $Q=0$. 

Having described the $Q$-graph, we now define the optimization variables. These variables correspond to the joint PMF $\cS \times \cU \times \cQ \times \cX \times \cS \times \cU \times \cQ$, where $\cU=\cQ$, and they are distributed according to \eqref{eq:Qgraph_joint_dist}. The policy $P_{U^+|U,Q}$ is constrained by \eqref{eq:policy_first_nodes}–\eqref{eq:policy_last_node}, and the strategy function is specified in \eqref{eq:x_behc}.
This policy structure, which is general for any given integer $N>0$ determining the $Q$-graph size, implies that the cardinalities $|\cU^{(Q=q)}|=q+1$ and $|\cU^{+(Q=q)}|=q+2$ increase with $q\in [0:N-1]$, and specifically $|\cU^{+(U=u,Q=q)}|=2$ for any $u$, while for the last node $|\cU^{(Q=N)}|=N+1, |\cU^{+(Q=N)}|=1$. These cardinalities are consistent with \eqref{eq:Uset_initialNodes}–\eqref{eq:Uset_LastNode}.
Based on the given joint PMF, we define the optimization variables as $\vv \mspace{-3mu}\triangleq\mspace{-3mu} P_{S,U,Q,X,S^+,U^+,Q^+}$, representing its non-zero elements.

Next, we define four sets of functions that serve as the constraints for the convex optimization problem.

\subsubsection{Stationary Distribution} 
Since $\vv$ has a stationary distribution over $(\cS, \cU, \cQ)$, the marginal distributions of $(S,U,Q)$ and $(S^+,U^+,Q^+)$ must be equal. That is, $P_{S,U,Q}(\tilde{s},\tilde{u},\tilde{q})=P_{S^+,U^+,Q^+}(\tilde{s},\tilde{u},\tilde{q})$ for all $(\tilde{s},\tilde{u},\tilde{q})\in \cS \times \cU^{(Q=\tilde{q})} \times \cQ$. Thus, for each realization $(\tilde{s},\tilde{u},\tilde{q})$, we define a constraint function:
\begin{align}\label{eq:constraints_stationary}
&f_i(\vv)  \triangleq \mspace{-10mu} \sum\limits_{s^+,u^+,q^+}\mspace{-20mu} P_{S,U,Q,X,S^+,U^+,Q^+}(\tilde{s},\mspace{-3mu}\tilde{u},\mspace{-3mu}\tilde{q},\mspace{-3mu}f(u^+,\mspace{-3mu}\tilde{s}),\mspace{-3mu}s^+,\mspace{-3mu}u^+,\mspace{-3mu}q^+) \nn\\ &\quad -\sum\nolimits_{s,u,q}P_{S,U,Q,X,S^+,U^+,Q^+}(s,\mspace{-3mu}u,\mspace{-3mu}q,\mspace{-3mu}f(\tilde{u},\mspace{-3mu}s),\mspace{-3mu}\tilde{s},\mspace{-3mu}\tilde{u},\mspace{-3mu}\tilde{q}),
\end{align}
where $i = 1,\dots,|\cS|\mspace{-5mu} \left(1+2+\cdots+(N+1)\right)\mspace{-5mu}=\mspace{-5mu}(N+1)(N+2)$ indexes all realizations. Note that the constraint functions in \eqref{eq:constraints_stationary} are linear in $\vv$.

\subsubsection{FSC Law}
The following constraint functions set
enforces that $\vv$ adheres to the FSC $P_{S^+|X,S}$ Markov chain $S^+-(X,S)-(U,Q,U^+)$, implied by \eqref{eq:Qgraph_joint_dist}, i.e., 
\begin{align}
&P\left(s,u,q,X=f(u^+,s),s^+,u^+\right)=P(s,u,q,u^+) \nn\\
&\quad \times P_{S^+|X,S}\left(s^+|f(u^+,s),s\right), \quad \forall s,u,q,s^+,u^+, \nn
\end{align}
Correspondingly, the constraint functions are defined as
\begin{align}\label{eq:constraints_channelLaw}
&f_i(\vv) \triangleq 
P\left(s,u,q,f(u^+,s),s^+,u^+,g(q,f(u^+,s))\right)-\\ 
& \sum\nolimits_{\tilde{s}^+} \mspace{-8mu} P\big(s,\mspace{-3mu}u,\mspace{-3mu}q,\mspace{-3mu}f\mspace{-2mu}(u^+\mspace{-3mu},\mspace{-3mu}s),\mspace{-3mu}\tilde{s}^+\mspace{-3mu}\mspace{-3mu},\mspace{-3mu}u^+\mspace{-3mu},\mspace{-3mu}g(q,\mspace{-3mu}f\mspace{-2mu}(u^+,s))\big) \mspace{-3mu} P\mspace{-2mu}\big(s^+\mspace{-2mu}|f\mspace{-2mu}(u^+\mspace{-3mu},\mspace{-3mu}s),\mspace{-3mu}s\big), \nn
\end{align}
for $i \mspace{-3mu}= \mspace{-3mu}(N\mspace{-3mu}+\mspace{-3mu}1)(N\mspace{-3mu}+\mspace{-3mu}2)\mspace{-3mu}+\mspace{-3mu}1,\dots,(N\mspace{-3mu}+\mspace{-3mu}1)(N\mspace{-3mu}+\mspace{-3mu}2)\mspace{-3mu}+\mspace{-3mu}2(N\mspace{-3mu}+\mspace{-3mu}1)(2N\mspace{-3mu}+\mspace{-3mu}1)\\=\mspace{-3mu}(N\mspace{-3mu}+\mspace{-3mu}1)(5N\mspace{-3mu}+\mspace{-3mu}4)$ that indexes all realizations. Note that $P_{S^+|X,S}$ remains constant for any choice of $\vv$, as defined in~\eqref{eq:stateEvolutionEH}.
Hence, the constraint functions in \eqref{eq:constraints_channelLaw} are linear in~$\vv$.

\subsubsection{Policy} 
The following constraint functions set enforces that $\vv$ adheres to the policy Markov chain $U^+-(U,Q)-S$ implied by \eqref{eq:Qgraph_joint_dist}, which can also be written as 
\begin{align}
{P(s,\mspace{-3mu}u,\mspace{-3mu}q,\mspace{-3mu}u^+\mspace{-3mu})}/{P(u,\mspace{-3mu}q,\mspace{-3mu}u^+\mspace{-3mu})}\mspace{-3mu}=\mspace{-3mu}P(s|u^+\mspace{-3mu},\mspace{-3mu}u,\mspace{-3mu}q)\mspace{-3mu}=\mspace{-3mu}\pi (s|u,\mspace{-3mu}q), \, \forall s,\mspace{-3mu}u,\mspace{-3mu}q,\mspace{-3mu}u^+\mspace{-3mu}. \nn   
\end{align} 
Since $|\cS|=2$ and $\pi(S=0|u,q)=1\mspace{-3mu}-\mspace{-3mu}\pi(S=1|u,q)$, the Markov chain holds if and only if $P(s=0,u,q,u^+)=\pi(S=0|u,q)P(u,q,u^+)$ for all $(u,q,u^+)$.
The corresponding constraint functions are given by
{
\setlength{\jot}{0pt}
\begin{align}\label{eq:constraints_policy}
& \mspace{-15mu} f_i(\vv) \triangleq \mspace{-5mu}
 \sum\nolimits_{\tilde{s}^+} \mspace{-5mu}P\mspace{-5mu}\left(s, \mspace{-3mu} u, \mspace{-3mu} q, \mspace{-3mu} f(u^+, \mspace{-3mu} s),\tilde{s}^+, \mspace{-3mu} u^+, \mspace{-3mu} g(q, \mspace{-3mu} f(u^+, \mspace{-3mu} s))\right)\mspace{-3mu}- \\&\mspace{-5mu}\sum\nolimits_{\tilde{s},\tilde{s}^+}\mspace{-5mu} P\mspace{-5mu} \left(\tilde{s},u,q,f(u^+,s),\tilde{s}^+,u^+,g(q,f(u^+,s))\right) \pi(s|u,q), \nn
\end{align}
}
where $s=0$, and the realization index $i$ spans from $i = (N\mspace{-3mu}+\mspace{-3mu}1)(5N\mspace{-3mu}+\mspace{-3mu}4)\mspace{-3mu}+\mspace{-3mu}1,\dots,(N\mspace{-3mu}+\mspace{-3mu}1)(5N\mspace{-3mu}+\mspace{-3mu}4)\mspace{-3mu}+\mspace{-3mu}(N\mspace{-3mu}+\mspace{-3mu}1)^2
  =(N\mspace{-3mu}+\mspace{-3mu}1)(6N\mspace{-3mu}+\mspace{-3mu}5)$. Since $\pi_{S|U,Q}$ remains constant for any choice of $\vv$, as given in \eqref{eq:s_given_uq_lb}, the constraint functions in \eqref{eq:constraints_policy} are linear in $\vv$.

\subsubsection{PMF} The final constraint function enforces that the PMF of $\vv$ is valid, i.e.,
\begin{align}\mspace{-3mu}
\label{eq:constraint_PMF}
  f_K(\vv)\mspace{-5mu}&\triangleq \mspace{-25mu} \sum_{s,u,q,s^+,u^+}\mspace{-25mu} P(s,\mspace{-3mu}u,\mspace{-3mu}q,\mspace{-3mu}f(u^+,\mspace{-3mu}s),\mspace{-3mu}s^+,\mspace{-3mu}u^+,\mspace{-3mu}g(q,\mspace{-3mu}f(u^+,\mspace{-3mu}s))\mspace{-3mu}  -\mspace{-3mu}  1,
\end{align}
where $K\mspace{-4mu} \triangleq \mspace{-4mu} (N\mspace{-4mu}+\mspace{-4mu}1)(6N\mspace{-4mu}+\mspace{-4mu}5)\mspace{-4mu}+\mspace{-4mu} 1\mspace{-4mu}=\mspace{-4mu}6N^2\mspace{-4mu}+\mspace{-4mu}11N\mspace{-4mu}+\mspace{-4mu}6$. Note that the constraint function in \eqref{eq:constraint_PMF} is linear $\vv$.

Using the constraint functions in \eqref{eq:constraints_stationary}–\eqref{eq:constraint_PMF}, we formulate the following convex optimization problem for lower bounding $C_{\text{BEHC}}$, consisting of $K$ linear constraints. 

\begin{tcolorbox}[colframe=black,colback=white, sharp corners,colbacktitle=white,coltitle=black,boxrule=0.45pt]
\vspace{-2mm}
\underline{Convex optimization problem formulation for the lower }\\\underline{bound sequence on the capacity of the BEHC, $C_{\text{BEHC}}$:}
\begin{align}\label{optProb_EH_LB}
& \underset{\vv}{\text{maximize}} &f_0(\vv)& \triangleq I(U^{+(N)},U^{(N)};X \mid Q^{(N)}) \nn\\
& \quad \quad \quad \text{s.t.} &f_i(\vv) & = \underline{0}, \; i = 1,\dots,K, \nn\\
& & \vv & \succeq \underline{0},
\end{align}
where $f_i(\vv)$ are detailed in \eqref{eq:constraints_stationary}–\eqref{eq:constraint_PMF}.
\vspace{-2mm}
\end{tcolorbox}

\begin{lemma}
\label{lemma:eh_lb_step_2}
For any integer $N\ge0$, Problem~\eqref{optProb_EH_LB}, corresponding to the $N+1$-node $Q$-graph for the lower bound, illustrated in Fig.~\ref{fig:EH_LB_Graph}, is a convex optimization problem. 
\end{lemma}

\begin{proof}
Since all constraint functions $f_i(\vv), i\in[1:K]$ are linear, it remains to show that the objective function $f_0(\vv)=H(X|Q)-H(X|U^+,U,Q)$ is concave in~$\vv$.

On the one hand, $-H(X|Q)\mspace{-4mu}=\mspace{-4mu}\sum_{x,q} P(q,x) \log \frac{P(q,x)}{P(q)}$ is the relative entropy, which is convex. Hence, $H(X|Q)$ is concave. On the other hand, $\mspace{-4mu}H(X\mspace{-2mu}|U^+\mspace{-6mu}=\mspace{-6mu}u^+,U\mspace{-6mu}=\mspace{-6mu}u,Q\mspace{-6mu}=\mspace{-6mu}q)$ is constant because $P(x|u^+,u,q)\mspace{-4mu}=\mspace{-4mu}\sum_{s} \pi(s|u,q) \mathbbm{1}\{x\mspace{-4mu}=\mspace{-4mu}f(u^+,s)\}$,
where $\pi(s|u,q)$ is constant, as given in \eqref{eq:s_given_uq_lb} and proved in Section~\ref{sec:QgraphLB}. This implies that $H(X|U^+,U,Q)$ is linear in $\vv$. Therefore, the objective function $f_0(\vv)$ is concave in $\vv$.
\end{proof}

\subsection{Convex Upper Bound on $C_{\text{BEHC}}$}
\label{subsec:EHconvex_upper}
The convex optimization problem for the upper bound closely resembles that of the lower bound in~\eqref{optProb_EH_LB}, but notable differences arise in the handling of the last node $Q=N$ in the $Q$-graph depicted in Fig.~\ref{fig:EH_UB_Graph}. Specifically, the outgoing edge $'X=0'$ from node $Q=N$ loops back to itself. The optimization variables are now the non-zero elements of the joint PMF $P_{S,U,Q,X,S^+,U^+,Q^+}$ in \eqref{eq:Qgraph_joint_distUB}, denoted by $\tilde{\vv}$.
Recall that this joint distribution includes $P_{S^+|X,S,Q}$ instead of $P_{S^+|X,S}$. The strategy function follows \eqref{eq:x_behc}, and the policy $P_{U^+|U,Q}$ satisfies \eqref{eq:policy_first_nodes}, while also adhering to \eqref{eq:policy_last_node_UB} instead of \eqref{eq:policy_last_node}. This general policy structure, applicable for any given $N$, induces the same auxiliary RV sets as in the lower bound, except that $|\cU^{+(Q=N)}|=2$. Moreover, $S \mid Q=N$ is always $1$, resulting in a cardinality of $1$.

These differences result in modifications to the FSC law constraints \eqref{eq:constraints_channelLaw} and the policy constraints \eqref{eq:constraints_policy}, as detailed below.
While the stationary distribution constraints \eqref{eq:constraints_stationary} and the PMF constraint \eqref{eq:constraint_PMF} remain unchanged (with $\tilde{\vv}$ replacing $\vv$), their indexing is updated.
The stationary distribution constraints now span up to $(N+1)^2$, and the PMF constraint indexing is adjusted based on the indexing of the following modified constraints.

\subsubsection{Modified FSC Law Constraint Functions}
The following constraint functions set
endorces that $\tilde{\vv}$ adheres to the Markov chain $S^+-(X,S,Q)-(U,U^+)$ implied by \eqref{eq:Qgraph_joint_distUB}. Accordingly, the FSC law constraint functions in \eqref{eq:constraints_channelLaw} are modified by replacing $P_{S^+|X,S}$ with $P_{S^+|X,S,Q}$. The redefined constraints are expressed as
\begin{align}\label{eq:constraints_channelLaw_modified}
&f_i(\tilde{\vv}) \triangleq 
P\left(s,u,q,f(u^+,s),s^+,u^+,g(q,f(u^+,s))\right)-\\ 
& \sum\nolimits_{\tilde{s}^+} \mspace{-10mu} P\mspace{-2mu} \big(s,\mspace{-4mu}u,\mspace{-4mu}q,\mspace{-4mu}f(u^+\mspace{-4mu},\mspace{-4mu}s),\mspace{-4mu}\tilde{s}^+\mspace{-4mu}\mspace{-4mu},\mspace{-4mu}u^+\mspace{-4mu},\mspace{-4mu}g(q,\mspace{-4mu}f(u^+,s))\big) P\mspace{-2mu}\big(s^+\mspace{-2mu}|f(u^+\mspace{-4mu},\mspace{-4mu}s),\mspace{-4mu}s,\mspace{-4mu}q\big), \nn
\end{align}
for $i \mspace{-3mu}=\mspace{-3mu} (N\mspace{-3mu}+\mspace{-3mu}1)^2\mspace{-3mu}+\mspace{-3mu}1,\dots,\mspace{-3mu} (N\mspace{-3mu}+\mspace{-3mu}1)^2\mspace{-3mu}+4(N\mspace{-3mu}+\mspace{-3mu}1)^2 =5(N\mspace{-3mu}+\mspace{-3mu}1)^2$, indexing all $(s,u,q,s^+,u^+)$. Note that $P_{S^+|X,S,Q}$ remains constant for any choice of $\tilde{\vv}$, as given in \eqref{eq:modifiedChannelLaw}. Hence, the modified constraint functions in \eqref{eq:constraints_channelLaw_modified} are linear in~$\tilde{\vv}$.

\subsubsection{Modified Policy Constraint Functions}
The constraint functions for the policy resemble those in \eqref{eq:constraints_policy}, but with
$\pi_{S|U,Q}$ given by \eqref{eq:s_given_uq_upper} instead of \eqref{eq:s_given_uq_lb}. Furthermore, the realization index $i$ spans from 
$i \mspace{-3mu}=\mspace{-3mu} 5(N\mspace{-3mu}+\mspace{-3mu}1)^2\mspace{-3mu}+\mspace{-3mu}1,\dots,5(N\mspace{-3mu}+\mspace{-3mu}1)^2\mspace{-3mu}+\mspace{-3mu}(N\mspace{-3mu}+\mspace{-3mu}1)(N\mspace{-3mu}+\mspace{-3mu}2)=(N\mspace{-3mu}+\mspace{-3mu}1)(6N\mspace{-3mu}+\mspace{-3mu}7)$, since there are $N\mspace{-3mu}+\mspace{-3mu}1$ additional realizations due to $|\cU^{+(Q=N)}|\mspace{-3mu}=\mspace{-3mu}2$. 
The modified $\pi_{S|U,Q}$ remains constant for any choice of ${\vv}$; thus, the modified policy constraint functions are linear in~$\tilde{\vv}$.

Finally, $\tilde{K}\triangleq (N\mspace{-3mu}+\mspace{-3mu}1)(6N\mspace{-3mu}+\mspace{-3mu}7)\mspace{-3mu}+\mspace{-3mu}1=6N^2\mspace{-3mu}+\mspace{-3mu}13N\mspace{-3mu}+\mspace{-3mu}8$ indexes the last constraint function (PMF) and replaces $K$. Incorporating these modifications, we formulate the following convex optimization problem to establish a sequence of computable upper bounds on $C_{\text{BEHC}}$.

\begin{tcolorbox}[colframe=black,colback=white, sharp corners,colbacktitle=white,coltitle=black,boxrule=0.45pt]
\vspace{-2mm}
\underline{Convex optimization problem formulation for the upper}\\
\underline{bound sequence on the capacity of the BEHC, $C_{\text{BEHC}}$:}
\begin{align}\label{optProb_EH_UB}
& \underset{\tilde{\vv}}{\text{maximize}} &f_0(\tilde{\vv})& \triangleq I(U^{+(N)},U^{(N)};X \mid Q^{(N)}) \nn\\
& \quad \quad \quad \text{s.t.} &f_i(\tilde{\vv}) & = \underline{0}, \; i = 1,\dots,\tilde{K}, \nn\\
& & \tilde{\vv} & \succeq \underline{0},
\end{align}
where $f_i(\tilde{\vv})$ are detailed in \eqref{eq:constraints_stationary}, 
\eqref{eq:constraints_channelLaw_modified}, \eqref{eq:constraints_policy}, and \eqref{eq:constraint_PMF}.
\vspace{-2mm}
\end{tcolorbox}

\begin{lemma}
\label{lem:BEHC_convex_ub_formulation}
For any integer $N\ge0$, Problem~\eqref{optProb_EH_UB}, corresponding to the $N+1$-node $Q$-graph for the upper bound, illustrated in Fig.~\ref{fig:EH_UB_Graph}, is a convex optimization problem. 
\end{lemma}

The proof of Lemma~\ref{lem:BEHC_convex_ub_formulation} follows identically to that of Lemma~\ref{lemma:eh_lb_step_2},
with the sole distinction that $\pi(s|u,q)$ is constant, as given by \eqref{eq:s_given_uq_upper} rather than \eqref{eq:s_given_uq_lb}, which is proved in Section~\ref{sec:QgraphUB}.


\section{Proofs of Main Results}
\label{sec:proof_main_theorems}
This section provides the proofs of Theorems~\ref{thr:EH_cvx_lower} and~\ref{thr:EH_cvx_upper}, which constitute the core contribution of this work.

\subsection{Proof of Theorem \ref{thr:EH_cvx_lower} -- Lower Bound on \(C_{\text{BEHC}}\)}
\label{sec:QgraphLB}
Theorem~\ref{thr:EH_cvx_lower} is based on Theorem~\ref{theorem:qgraph_LB} and is proved in four parts.
First, the policy's compliance with the BCJR constraints is established in Lemma~\ref{lemma:eh_lb_BCJR}.
Second, the marginal \(\pi_{S|U,Q}\) is shown to satisfy \eqref{eq:s_given_uq_lb}, as detailed in Lemma~\ref{lem:marginal_s_given_uq}. Third, the convexity of the lower bound \eqref{eq:EH_LB_cvx} is demonstrated in Lemma~\ref{lemma:eh_lb_step_2}, Section~\ref{subsec:EHconvex_lower}. Finally, the convergence of the lower bound sequence to \(C_{\text{BEHC}}\) is proved in Theorem~\ref{thr:convergence}, Section~\ref{subsec:ProofThrConvergence}, alongside the convergence of the upper bound sequence.

\begin{lemma}
\label{lemma:eh_lb_BCJR}    
Assuming the auxiliary RV sets given in \eqref{eq:Uset_initialNodes}–\eqref{eq:Uset_LastNode}, any policy satisfies the BCJR constraints \eqref{eq:BCJR_Def}.    
\end{lemma}
\begin{proof}
\label{appendix:ProofBCJR_CVX}
The BCJR constraint \eqref{eq:BCJR_Def} can alternatively be expressed as follows (see \cite[Eq.~12]{shemuel2024finite}):
\begin{align}
&\pi_{U^+,S^+|Q^+} (u^+,s^+|q^+) \label{eq:BCJR_check}\\
&=\frac{\sum_{u,s}\pi(u,s|q)P(u^+|u,q) \mathbbm{1}\{x = f(u^+,s)\} P(s^+|x,s)}{\sum_{u,u'^+,s} \pi(u,s|q)P(u'^+|u,q)\mathbbm{1}\{x = f(u'^+,s)\}}. \nn 
\end{align}
To prove the claim, we need to show that for any $(u^+,s^+, q^+)$, all $(q,x)$ satisfying $q^+ = g(q,x)$ yield the same result in the RHS of \eqref{eq:BCJR_check}. Thus, it suffices to verify this property for nodes with more than one incoming edge. In the $Q$-graph of Fig.~\ref{fig:EH_LB_Graph}, the only such node is $Q^+=0$, which arises in two cases:
\\\textit{Case (I):} $g(q,X=1)=0, \quad \forall q\in \cQ^{(N)}$.
\\\textit{Case (II):} $g(Q=N,X=0)=0$.\\
Since $\cU^{(Q=0)}=\{0\}, |\cS|=2$, it suffices to verify \eqref{eq:BCJR_check} for these two cases when $(u^+,s^+)=(0,0)$ only.

For Case (I), the RHS of \eqref{eq:BCJR_check} when $(u^+,s^+)=(0,0)$ is 
\begin{align}
\label{eq:BCJR_ev_eh1}
&\frac{\sum_{u,s}\mspace{-5mu}\pi(u,\mspace{-3mu}s|q)P(U^+\mspace{-6mu}=\mspace{-4mu}0|u,\mspace{-3mu}q)
\mspace{-3mu}\mathbbm{1}\mspace{-3mu}\{1 \mspace{-5mu}=\mspace{-5mu} f(0,\mspace{-3mu}s)\} P_{S^+|X,S}(0|1,\mspace{-3mu}s)}{\sum_{u,u'^+,s} \pi(u,s|q)P(U^+\mspace{-5mu}=\mspace{-5mu}u'^+|u,\mspace{-3mu}q)
\mathbbm{1}\{1 \mspace{-5mu}=\mspace{-5mu} f(u'^+\mspace{-3mu},\mspace{-3mu}s)} \nn\\
&\stackrel{(a)}=\frac{\sum_{u}\pi(u,S=1|q)P(U^+=0|u,q)P_{S^+|X,S}(0|1,1)}{\sum_{u} \pi(u,S=1|q)P(U^+=0|u,q)}
 \nn\\
&\stackrel{(b)}=\frac{\sum_{u}\pi(u,S=1|q)P(U^+=0|u,q)\bar{\eta}}{\sum_{u} \pi(u,S=1|q)P(U^+=0|u,q)}
 =\bar{\eta}.
\end{align}
Step (a) follows from \eqref{eq:x_behc}; 
Step (b) follows from \eqref{eq:stateEvolutionNoiselessEH}.

For Case (II), the RHS of \eqref{eq:BCJR_check} when $(u^+,s^+)=(0,0)$ is
\begin{align}
\label{eq:BCJR_ev_eh2}
& \frac{\sum_{u,s}\mspace{-4mu}\pi(u,\mspace{-3mu}s|q)\mspace{-2mu} P_{U^+|U,Q}(0|u,N)
\mathbbm{1}\mspace{-3mu}\{0 \mspace{-2mu}=\mspace{-5mu} f(0,\mspace{-3mu} s)\}P_{S^+|X,S}(0|0, \mspace{-3mu}s)}{\sum_{u,u'^+,s} \pi(u, \mspace{-3mu} s|q)P_{U^+|U,Q}(u'^+|u,\mspace{-3mu}N)
\mathbbm{1}\mspace{-3mu}\{0 \mspace{-3mu}=\mspace{-2mu} f(u'^+,\mspace{-3mu} s)\}
} \nn\\
& \stackrel{(a)}=\frac{\sum_{u,s}\mspace{-3mu} \pi(u,s|q)
\mathbbm{1}\{0 = f(0,s)\}
P_{S^+|X,S}(0|0,s)}{\sum_{u,s} \pi(u,s|q)\mathbbm{1}\{0 = f(0,s)\}}
 \nn\\ 
 & \stackrel{(b)} = \frac{\sum_{u}\pi(u,S=0|q)\mathbbm{1}\{0 = f(0,0)\} \bar{\eta}}{\sum_{u,s} \pi(u,s|q)
\mathbbm{1}\{0 = f(0,s)\}}\stackrel{(c)}=\bar{\eta}.
\end{align}
Step (a) follows from \eqref{eq:policy_last_node}; 
Step (b) follows from \eqref{eq:stateEvolutionNoiselessEH}; Step (c) follows from \eqref{eq:x_behc}, which implies that $x=s$ when $u^+=0$. Therefore, \eqref{eq:BCJR_ev_eh1} equals \eqref{eq:BCJR_ev_eh2}, completing the proof.
\end{proof}

\begin{lemma}
\label{lem:marginal_s_given_uq}
For any policy, assuming the auxiliary RV sets in \eqref{eq:Uset_initialNodes}–\eqref{eq:Uset_LastNode}, the marginal \(\pi_{S|U,Q}\) is constant and given by 
\begin{align}
\label{eq:s_given_uq_lb_induction}
\pi(S=0|u,q)=\bar{\eta}^{u+1}, \quad \forall q\in \cQ, \forall u\in \mathcal{U}^{(Q = q)}.
\end{align}
\end{lemma}
\begin{proof}
We prove the statement by induction.\\ 
\textit{Base case:} For all $q\in \cQ$, when $u=0$ we have:
\begin{align}
    &\pi_{S|U,Q}(0|0,q)\stackrel{(a)}=\pi_{S^+|U^+,Q^+}(0|0,q) \nn\\
    &=\mspace{-3mu}\sum\nolimits_{s,x \mid s\ge x} P_{S,X,S^+|U^+,Q^+}(0,s,x|0,q)\nn\\
    &=\mspace{-3mu}\sum\nolimits_{s,x \mid s\ge x} \mspace{-7mu} P_{S,X|U^+,Q^+}(s,x|0,q) P_{S^+|S,X}(0|s,f(0,s))\mspace{-3mu}\stackrel{(b)}=\mspace{-3mu}\bar{\eta}. \nn
\end{align}
Step (a) follows from the stationarity of the distribution over $(S,U,Q)$. Step (b) follows from \eqref{eq:stateEvolutionNoiselessEH} and \eqref{eq:x_behc}, which imply that $P_{S^+|S,X}(0|s,f(0,s))=\bar{\eta}$.\\
\textit{Inductive hypothesis:} Suppose the statement holds for all $u \in \mathcal{U}^{(Q = q)}\setminus \{0,q\}, q \in \cQ \setminus \{0,N\}$. I.e., $\pi(S\mspace{-2mu}=\mspace{-2mu}0|u,\mspace{-2mu}q)\mspace{-2mu}=\mspace{-2mu}\bar{\eta}^{u+1}$.\\
\textit{Inductive step:} For $\tilde{u}= u+1, \tilde{q}= q+1$, we show that $\pi_{S|U,Q}(0|\tilde{u},\tilde{q})=\bar{\eta}^{\tilde{u}+1}$. Expanding:
\begin{align}
    &\pi_{S|U,Q}(0|\tilde{u},\tilde{q})\stackrel{(a)}=\pi_{S^+|U^+,Q^+}(0|\tilde{u},\tilde{q}) \nn\\    
    &\stackrel{(b)}=P_{S^+|U^+,Q^+,U,Q,X}(0|\tilde{u},\tilde{q},\tilde{u}-1,\tilde{q}-1,0) \nn\\
    &\stackrel{(c)}=P_{S,S^+|U^+,Q^+,U,Q,X}(0,0|\tilde{u},\tilde{q},\tilde{u}-1,\tilde{q}-1,0) \nn\\
    &\stackrel{(d)}=P_{S|U,Q}(0|\tilde{u}-1,\tilde{q}-1)P_{S^+|X,S}(0|0,0) 
    \nn \\    &\stackrel{(e)}=\bar{\eta}^{\tilde{u}}\bar{\eta}=\bar{\eta}^{\tilde{u}+1}. \nn    
\end{align}
Step (a) follows from the stationary distribution.
Step (b) follows since when $\tilde{u}\ne 0$, it implies $X=f(\tilde{u},s)=0, \forall s$. Moreover, $U^+=\tilde{u}, Q^+=\tilde{q}$ are reachable only from $U=\tilde{u}-1, Q=\tilde{q}-1$. Step (c) follows because if $X=0, S^+=0$, it is impossible for $S=1$ due to \eqref{eq:stateEvolutionNoiselessEH}. Step (d) follows from the Markov property of $P_{S^+|X,S}$ and the Markov chain $S-(U,Q)-(U^+,Q^+,X)$. Its proof is omitted for brevity, but follows from  
\eqref{eq:Qgraph_joint_dist} and Bayes' rule. Step (e) follows from the inductive hypothesis and \eqref{eq:stateEvolutionNoiselessEH}. Thus, the statement is proven.
   
\end{proof}




\subsection{Proof of Theorem \ref{thr:EH_cvx_upper} -- Upper Bound on $C_{\text{BEHC}}$}
\label{sec:QgraphUB}
To establish the upper bound, we analyze a variation of the BEHC, whose capacity exceeds $C_{\text{BEHC}}$. Specifically, for any integer $N\ge0$, consider a modified BEHC where transmitting $N$ consecutive zero channel inputs automatically charges the battery, i.e., $P(S_i=1|X_{i-N}^{i-1}=\mathbf{0}^N)=1$.
We refer to this variation as the \textit{boosted BEHC($N$)}. 
Let $\bar{C}_{\text{BEHC}}(N)$ denote the capacity of the boosted BEHC($N$). Then, it follows that $C_{\text{BEHC}} < \bar{C}_{\text{BEHC}}(N) , \, \forall N\ge0$.
To prove Theorem~\ref{thr:EH_cvx_upper}, we demonstrate that the RHS of \eqref{eq:EH_UB_cvx} serves as an upper bound for $\bar{C}_{\text{BEHC}}(N)$ for any $N$.
The proof consists of three main steps, formulated in Lemma~\ref{lem:C_bar_FSC}, Theorem~\ref{thr:FSC_Qgraph_UB} and Theorem~\ref{thr:cardinality} below. 

Observe that the FSC property \eqref{eq:BasicFSCMarkov} with the original state space does not hold for the boosted BEHC($N$). For instance, consider the case $N\mspace{-3mu}=\mspace{-3mu}2$: 
\begin{align}
&P\mspace{-3mu}\left(S_i\mspace{-3mu}=\mspace{-3mu}1|(X_{i-1},\mspace{-3mu} X_{i})\mspace{-3mu}=\mspace{-3mu}(0,\mspace{-3mu} 0),\mspace{-3mu} (S_{i-2}, \mspace{-3mu} S_{i-1})\mspace{-3mu}=\mspace{-3mu}(0,\mspace{-3mu} 0)\right)\mspace{-3mu}=\mspace{-3mu}1 \nn\\
&> P(S_i\mspace{-3mu}=\mspace{-3mu}1|(X_{i-1},X_{i})\mspace{-3mu}=\mspace{-3mu}(1,0),(S_{i-2},S_{i-1})\mspace{-3mu}=\mspace{-3mu}(1,0)). \nn
\end{align}
This demonstrates a dependence on $X_{i-1}, S_{i-2}$, violating the FSC property. Nevertheless, in the first step, we show that the boosted BEHC($N$) can be reformulated as a FSC by introducing a new state definition.

\begin{lemma}[Step 1]
\label{lem:C_bar_FSC}
The boosted BEHC($N$) can be reformulated as a strongly connected FSC $P_{\tilde{S}^+,Y|X,\tilde{S}
}$ with a new state $\tilde{S}_i \triangleq (S_i,Q_i^{(N)}
)$ and initial state $(s_0,q_0)=(0,0)$, given by
\begin{align}
\label{eq:NewFSC}
    &P_{\tilde{S}^+,Y|X,\tilde{S}}((s_i,q_i),y_i|x_i,(s_{i-1},q_{i-1}))
    =\mathbbm{1}\{y_i=x_i\} \nn\\& \times P_{S^+|X,S,Q^{(N)}}(s_i|x_i,s_{i-1},q_{i-1}) \mathbbm{1}\{q_i=g(q_{i-1},x_i)\},
\end{align}
where $P_{S^+|X,S,Q^{(N)}}$ is given in \eqref{eq:modifiedChannelLaw},
and $g(q_{i-1},x_i)$ corresponds to the transitions of the $N\mspace{-3mu}+\mspace{-3mu}1$-node $Q$-graph in Fig.~\ref{fig:EH_UB_Graph}.
\end{lemma}

Having established that the boosted BEHC($N$) can be reformulated as a strongly connected FSC, the next step 
is to derive a single-letter $Q$-graph upper bound on $\bar{C}_{\text{BEHC}}(N)$, which is a specific instance of the general $C_{\text{fb-csi}}$, based on Fig.~\ref{fig:EH_UB_Graph}.
Notably, unlike the general $Q$-graph lower bound in Theorem~\ref{theorem:qgraph_LB}, no general $Q$-graph upper bound on $C_{\text{fb-csi}}$ has been established in \cite{shemuel2024finite}, primarily due to the lack of a cardinality bound on $\cU$. 

\begin{theorem}[Step 2]
\label{thr:FSC_Qgraph_UB}
The capacity of the boosted BEHC($N$) is upper bounded by
\begin{align}
\label{eq:lemmaOron}
    \bar{C}_{\text{BEHC}}(N) &\le  \sup I(U^{+(N)},U^{(N)};X|Q^{(N)}),
\end{align}
where the RV $Q^{(N)}$ is defined on the vertices of the $Q$-graph illustrated in Fig.~\ref{fig:EH_UB_Graph}, 
the joint distribution is given by \eqref{eq:Qgraph_joint_distUB}
in which $P_{S^+|X,S,Q^{(N)}}$ is defined by \eqref{eq:modifiedChannelLaw},
and the supremum is taken with respect to
${\{f: \mathcal U  \times \mathcal S \to \mathcal X \; | \; f(u^+,0)=0 \; \forall u^+ \}, \; P(u^+|u,q) \in \mathcal{P}_\pi}$. 
\end{theorem}

In the third step, we demonstrate that the supremum of \eqref{eq:lemmaOron} is achieved with $f(u^+,s)$ and the set of policies $P_{U^+|U,Q}$ specified in Theorem~\ref{thr:EH_cvx_upper} with a finite cardinality $|\cU|$. 

\begin{theorem}[Step 3]
\label{thr:cardinality}
The supremum in \eqref{eq:lemmaOron} can be achieved with $\cU=\cQ=[0:N]$, a strategy function $f(u^+,s)$ specified as in \eqref{eq:x_behc},
and a policy $P_{U^+|U,Q} \in \mathcal{P}_\pi$ subject to
\begin{align}     
    &P(u^+|u,q)=0, \quad \forall q\in [0: N-1],  \forall u\in [0:q], \nn\\ &\quad \forall u^+\in [1:q+1] \backslash \ (u+1), \label{eq:policy_first_nodes_proof}\\
    &P(u^+|u,Q=N)=0, \forall u\in [0:N], \forall u^+\in [2:N], \label{eq:policy_last_node_UB_proof}     
\end{align}
yielding the following auxiliary sets: 
\begin{align}
     &\mspace{-10mu}\cU^{(Q=q)}\mspace{-5mu}=\mspace{-5mu}[0\mspace{-2mu}:\mspace{-2mu}q], \cU^{+(U=u,Q=q)}\mspace{-5mu}=\mspace{-5mu}\{0,u\mspace{-5mu}+\mspace{-5mu}1\}, \forall q\mspace{-2mu}\in\mspace{-2mu}[0\mspace{-2mu}:\mspace{-2mu}N\mspace{-3mu}-\mspace{-3mu}1], \mspace{-12mu}\label{eq:Uset_initialNodesProof} \\
     &\mspace{-10mu} \cU^{(Q=N)}=[0:N], \; \cU^{+(Q=N)}=\{0,1\}. \label{eq:Uset_lastNodeProofUB}
\end{align}
Furthermore, the induced marginal $\pi_{S|U,Q}$ is given by \eqref{eq:s_given_uq_upper}.
\end{theorem} 
In conclusion, Steps~1-3 proved below establish the single-letter upper bound on $C_{\text{BEHC}}$ presented in Theorem~\ref{thr:EH_cvx_upper}. Moreover, Lemma~\ref{lem:BEHC_convex_ub_formulation}, shown in Section~\ref{subsec:EHconvex_upper}, formulates \eqref{eq:lemmaOron} as a convex optimization problem. Thus, after proving these three steps, the final task to complete the proof of Theorem~\ref{thr:EH_cvx_upper} is to establish its convergence to $C_{\text{BEHC}}$, as provided in Section~\ref{subsec:ProofThrConvergence}.

\subsubsection{Proof of Lemma~\ref{lem:C_bar_FSC} (Step 1)} 
In the new state $\tilde{S}_i\triangleq (S_i,Q_i^{(N)})$, the auxiliary RV $Q_i^{(N)}$, defined on the vertices of the $Q$-graph in Fig.~\ref{fig:EH_UB_Graph},
maps the previous $N$ inputs to a final node. Specifically, $Q_{i-1}: \cX^N \to \cQ$, where $\cQ=[0:N]$ and $Q_{i-1}=g(x_{i-N}^{i-1})$. The function $g(x_{i-N}^{i-1})$ determines the resulting node as the length of consecutive $0'$s following the last $'1'$ in the previous $N$ inputs, with $g(\mathbf{0}^N)=N$. Thus, 
\begin{align}
    &P(\tilde{s}_i,y_i|x^i,\tilde{s}_0^{i-1},y^{i-1},m)\nn\\
    &=P\left((s_i,q_i),y_i|x^i,(s_0^{i-1},q_0^{i-1}(x^{i-1})),y^{i-1},m\right) \nn\\
    &= P_{S^+|X,S,Q^{(N)}}\left(s_i|x_i,s_{i-1},q_{i-1}(x_{i-N}^{i-1})\right ) \nn\\ &\quad \times \mathbbm{1}\{q_i=g\left(q_{i-1}(x_{i-N}^{i-1}),x_i\right)\}  \mathbbm{1}\{y_i=x_i\} \nn\\
    &= P_{\tilde{S}^+,Y|X,\tilde{S}}(\tilde{s}_i,y_i|x_i,\tilde{s}_{i-1}),
\end{align}
where $P_{S^+|X,S,Q^{(N)}}$ is given in \eqref{eq:modifiedChannelLaw}.
This demonstrates that the FSC property~\eqref{eq:BasicFSCMarkov} holds with the new state. 
The connectivity property in Definition~\eqref{def:connectivity} is satisfied because any initial state $\tilde{s}'$ can be driven to any desired state $\tilde{s}$ with positive probability by transmitting the appropriate length of consecutive zeros. In conclusion, the boosted BEHC($N$) can be formulated as a strongly connected FSC with state $\tilde{S}$. \qed

Steps $2$ and $3$ incorporate an auxiliary RV 
$K_{i-1}^{(N)}\mspace{-3mu}\in\mspace{-3mu}[i\mspace{-3mu}-\mspace{-3mu}N\mspace{-3mu}:\mspace{-3mu}i\mspace{-3mu}-\mspace{-3mu}1]$ for any $i \ge N$, defined as the index of the last input '$1$' among the previous $N$ inputs, $X_{i-N}^{i-1}$, if such an input exists (\textit{Case $(i)$}), i.e., $X_{K_{i-1}}^{i-1}\mspace{-3mu}=\mspace{-3mu}(1,\mathbf{0}^{i-1-{K_{i-1}}})$. Otherwise, 
if $X_{i-N}^{i-1}=\mathbf{0}^N$ (\textit{Case $(ii)$}), we define $K_{i-1}\mspace{-3mu}\triangleq\mspace{-3mu} i\mspace{-3mu}-\mspace{-3mu}(N\mspace{-3mu}+\mspace{-3mu}1)$. 
For simplicity, we omit the superscript $^{(N)}$ and refer to $K_{i-1}$, while recalling that it depends on $N$.
Notice that there is a one-to-one mapping between $K_{i-1}$ and $Q_{i-1}$, given by $K_{i-1}\mspace{-3mu}=\mspace{-3mu}i\mspace{-3mu}-\mspace{-3mu}(Q_{i-1}\mspace{-3mu}+\mspace{-3mu}1)$.

We now proceed to prove Step 2.

\subsubsection{Proof of Theorem~\ref{thr:FSC_Qgraph_UB} (Step 2)}
Our starting point is the converse result in \cite[Eq.~(42)]{shemuel2024finite}, which states that given a fixed sequence of $(2^{nR},n)$ codes for a FSC $P_{S^+,Y|X,S}$ with feedback and SI available causally to the encoder, as depicted in Fig~\ref{fig:setting}, the achievable rate $R$ must satisfy
\begin{align}
\label{eq:UBepsilonN}
R-\epsilon_n&\le \frac{1}{n} \max_{P(u^n||y^{n-1})}
 I(U^n \to Y^n),
\end{align}
where $\epsilon_n\to 0$ as $n \to \infty$, $U_i$ indexes all possible strategies $\{f_u (s):\cS \to \cX \}$, and the joint distribution is given by
\begin{align}
    \label{eq:joint_distDI_prf}
    &P(s_0^n,u^n,x^n,y^n)=P(s_0)P(u^n||y^{n-1}) \nn\\
    &\times \prod\nolimits_{i=1}^n \mspace{-6mu}\mathbbm{1}\{x_i \mspace{-3mu}=\mspace{-3mu} 
    f_{u_i}(s_{i-1})\}P_{S^+,Y|X,S}(s_i,\mspace{-3mu} y_i|x_i,\mspace{-3mu} s_{i-1})
    \text{.}
\end{align}
By Lemma~\ref{lem:C_bar_FSC}, the boosted BEHC($N$) is a FSC $P_{\tilde{S}^+|X,\tilde{S}}$ with state $\tilde{S}_{i-1}=(S_{i-1},Q^{(N)}_{i-1})$. Applying \eqref{eq:UBepsilonN} to this FSC with the joint distribution \eqref{eq:joint_distDI_prf}, where 
$Y=X$ and $\tilde{S}_{i-1}$ is used instead of $S_{i-1}$, we obtain the joint distribution
\begin{align}    
    &P(\tilde{s}_0^n,u^n,x^n)=P_{\tilde{S}_0}(\tilde{s}_0)P(u^n||x^{n-1}) \nn\\
    &\quad \times \prod\nolimits_{i=1}^n\mathbbm{1}\{x_i = \tilde{f}_{u_i}(\tilde{s}_{i-1})\}  P_{\tilde{S}^+|X,\tilde{S}}(\tilde{s}_i|x_i,\tilde{s}_{i-1}) \nn\\
    &\stackrel{(*)}=P_{S_0,Q^{(N)}_0}(s_0,q_0) P(\hat{u}^n||x^{n-1}) \prod\nolimits_{i=1}^n \mspace{-6mu}\mathbbm{1}\{x_i = f_{\hat{u}_i}(s_{i-1})\} \nn\\ & \times P_{S^+|X,S,Q^{(N)}}(s_i|x_i,s_{i-1},q_{i-1})  \mathbbm{1}\{q_i=g(q_{i-1},x_i)\}
    \text{.} \label{eq:joint_boosted}
\end{align}
Here, the mappings $\tilde{f}:\tilde{\cS} \to \cX$ are subject to the battery constraint, i.e., $\tilde{f}_{u_i}(s_{i-1}=0,q_{i-1})=0 \;\forall u_i, q_{i-1}$. Step ($*$) holds because $x_i = \tilde{f}_{u_i}(s_{i-1},q_{i-1})$ can be rewritten as a time-invariant function $f_{\hat{u}_i}(s_{i-1})$ with another auxiliary RV $\hat{u}_i\triangleq(u_i,q_{i-1})$, preserving the causal conditioning and the objective at any time $i$:
\begin{align}
    I(U^i;Y_i|X^{i-1})\mspace{-3mu}&=\mspace{-3mu}I(U^i,Q^{i-1};Y_i|X^{i-1})\mspace{-3mu}=\mspace{-3mu}I(\hat{U}^i;Y_i|X^{i-1}), \nn\\
    P(u_i|u^{i-1},\mspace{-3mu}x^{i-1})\mspace{-3mu}&=\mspace{-3mu}P(u_i|u^{i-1},\mspace{-3mu}q^{i-2},\mspace{-3mu}x^{i-1}) \mathbbm{1}\mspace{-3mu}\{q_{i-1}\mspace{-3mu}=\mspace{-3mu}g(x_{i-N}^{i-1})\}\nn\\
    \mspace{-3mu}&=\mspace{-3mu}P(\hat{u}_i|\hat{u}^{i-1},\mspace{-3mu}x^{i-1}). \nn
\end{align}
This follows because $q_{i-1}$ is a deterministic function of $x^{i-1}_{i-N}$. Furthermore, the Markov chain $\hat{U}_i-(\hat{U}^{i-1},Y^{i-1})-S^{i-1}$ holds similarly to the original Markov chain involving $U_i$. Consequently, considering the battery constraint ($X_i \le S_{i-1}$), \eqref{eq:joint_boosted} implies that any achievable rate must satisfy
\begin{align}
\label{eq:UBepsilonN_boosted}
R-\epsilon_n&\le \frac{1}{n} \max_{P(\hat{u}^n||x^{n-1})} I(\hat{U}^n \to X^n),
\end{align}
where the mapping $f_{\hat{u}^+}(s)$ is constrained by $f_{\hat{u}^+}(0)=0, \forall \hat{u}^+$.

The remainder of the proof comprises two parts.
In the first part, we upper bound the objective in \eqref{eq:UBepsilonN_boosted} by a single-letter expression that depends on $n$. In the second part, we take the limit as $n\to \infty$ to derive the single-letter upper bound in~\eqref{eq:lemmaOron}. 

For the first part, consider the directed information in \eqref{eq:UBepsilonN_boosted}:
\begin{align}
\label{eq:step_EH}
    &\frac{1}{n} \sum\nolimits_{i=1}^n I(\hat{U}^i;X_i|X^{i-1})  \nn\\
    &\stackrel{(a)}\le \frac{1}{n} \sum\nolimits_{i=1}^n H(X_i|Q_{i-1}^{(N)})-H(X_i|\hat{U}^i,X^{i-1},Q_{i-1}^{(N)})
    \nn\\
    &\stackrel{(b)}= \frac{1}{n} \sum\nolimits_{i=1}^n H(X_i|Q_{i-1}^{(N)})-H(X_i|\hat{U}_{K_{i-1}^{(N)}}^i,Q^{(N)}_{i-1})
    \nn\\     
    &\stackrel{(c)}= \frac{1}{n} \sum\nolimits_{i=1}^n I(\tilde{U}^{(N)}_i,\tilde{U}^{(N)}_{i-1};X_i|Q_{i-1}^{(N)})
    \nn\\     
    &\stackrel{(d)}\le I(\tilde{U}^{+(N)},\tilde{U}^{(N)};X|Q^{(N)}) \nn\\
    &\stackrel{(e)}\le \max_{P \in \mathcal{D}_{\frac{1}{n}}}  I(\tilde{U}^{+(N)},\tilde{U}^{(N)};X|Q^{(N)}).
\end{align}
Step (a) follows from $Q_{i-1}\triangleq g(X_{i-N}^{i-1})$ as defined by the $Q$-graph in Fig.~\ref{fig:EH_UB_Graph}, and because conditioning reduces entropy. 
Step (b) follows from Lemma~\ref{lem:boosted_markov} in Appendix~\ref{appendix:lemma_boosted}, which establishes the Markov chain $X_i-(\hat{U}_{K_{i-1}}^i,Q_{i-1})-(\hat{U}^{K_{i-1}-1},X^{i-1})$.
Step (c) introduces the auxiliary RV $\tilde{U}^{(N)}_i \mspace{-3mu} \triangleq \mspace{-3mu}\hat{U}_{K_{i-1}+1}^i$ and similarly $\tilde{U}^{(N)}_{i-1} \mspace{-3mu} \triangleq \mspace{-3mu}\hat{U}_{K_{i-1}}^{i-1}$. 
We revisit Step (c) in the next theorem by redefining $\tilde{U}_i^{(N)}$ with the auxiliary RV sets $\cU^{(Q=q)}, q\in \cQ$ given in \eqref{eq:Uset_initialNodes}, \eqref{eq:Uset_LastNodeUB}, such that the joint distribution in Step (d) is in the form of \eqref{eq:Qgraph_joint_distUB}. 
Step (d) applies Jensen's inequality, leveraging the concavity of $I(\tilde{U}^+,\tilde{U};X|Q)$ with respect to the joint distribution $P_{S,\tilde{U},Q,X,S^+,\tilde{U}^+,Q^+}$ (see Lemma~\ref{lem:BEHC_convex_ub_formulation}, Section~\ref{subsec:EHconvex_upper}) given by 
\begin{align}
&P(s,\tilde{u},q,x,s^+,\tilde{u}^+,q^+)=\overline{P}(s,\tilde{u},q)\overline{P}(\tilde{u}^+|\tilde{u},q)\nn\\ &\times \mspace{-3mu}\mathbbm{1}\mspace{-2mu}\{x\mspace{-3mu}=\mspace{-3mu}f(\tilde{u}^+\mspace{-4mu},\mspace{-3mu}s')\} P_{S^+|X,S,Q}(s^+\mspace{-3mu}|x,\mspace{-3mu}s,\mspace{-3mu}q) \mathbbm{1}\mspace{-2mu}\{q^+\mspace{-3mu}=\mspace{-3mu}g(q,\mspace{-3mu}x)\}, \nn\\
&\overline{P}(s,\tilde{u},q) \triangleq \frac{1}{n} \sum\nolimits_{i=1}^n P_{S_{i-1},\tilde{U}_{i-1},Q_{i-1}}(s,\tilde{u},q), \label{eq:averageDist_EH}\\
&\overline{P}(\tilde{u}^+|\tilde{u},q) \triangleq \frac{1}{n} \sum\nolimits_{i=1}^n P_{\tilde{U}_{i}|\tilde{U}_{i-1},Q_{i-1}}(\tilde{u}^+|\tilde{u},q).  \label{eq:averagePolicyDist_EH}
\end{align}
Step (e) defines $\mathcal{D}_{\epsilon}$, the set
of distributions satisfying
\begin{align}
&\mathcal{D}_{\epsilon} \triangleq \{ P_{S,\tilde{U},Q}(s,\tilde{u},q)P_{\tilde{U}^+|\tilde{U},Q}(\tilde{u}^+|\tilde{u},q) \in \mathcal{P}_{\cS \times \tilde{\cU} \times \cQ \times \tilde{\cU}} :\nn\\
&|P_{S,\tilde{U},Q}(s,\mspace{-4mu}\tilde{u},\mspace{-4mu}q) \mspace{-4mu}-\mspace{-4mu} \sum\nolimits_{s',\tilde{u}',q',x} \mspace{-15mu} P_{S,\tilde{U},Q}(s'\mspace{-2mu},\mspace{-4mu}\tilde{u}',\mspace{-4mu}q')P_{\tilde{U}^+|\tilde{U},Q}(\tilde{u}|\tilde{u}'\mspace{-3mu},\mspace{-4mu}q') \nn\\
&\mspace{-4mu}\times \mspace{-6mu} \mathbbm{1}\mspace{-3mu}\{x\mspace{-5mu}=\mspace{-5mu}f(\mspace{-2mu}\tilde{u},\mspace{-4mu}s')\} \mspace{-2mu}P_{S^+|X,S,Q}(s|x,\mspace{-4mu}s',\mspace{-4mu}q')  \mspace{-3mu} \mathbbm{1}\mspace{-3mu}\{q\mspace{-4mu}=\mspace{-4mu}g(\mspace{-2mu}q'\mspace{-3mu},\mspace{-4mu}x)\} |\mspace{-4mu}\le\mspace{-4mu} \epsilon, \mspace{-2mu}\forall s,\mspace{-4mu}\tilde{u},\mspace{-4mu}q \}. \nn
\end{align}
For any codebook of length $n$, its induced averaged distribution $\overline{P}(s,\tilde{u},q,\tilde{u}^+)=\overline{P}(s,\tilde{u},q)\overline{P}(\tilde{u}^+|\tilde{u},q)$ lies in $\mathcal{D}_{\frac{1}{n}}$, as shown by the following derivation
using Definitions \eqref{eq:averageDist_EH}–\eqref{eq:averagePolicyDist_EH}:
\begin{align}
    &\Big|\overline{P}_{S,\tilde{U},Q}(s,\mspace{-3mu}\tilde{u},\mspace{-3mu}q) \mspace{-4mu}-\mspace{-5mu} \sum\nolimits_{s',\mspace{-1mu}\tilde{u}',\mspace{-1mu}q',\mspace{-1mu}x}\mspace{-5mu} \overline{P}_{S,\tilde{U},Q}\mspace{-2mu}(s',\mspace{-2mu}\tilde{u}',\mspace{-2mu}q') \overline{P}_{\tilde{U}^+|\tilde{U},Q}\mspace{-2mu}(\tilde{u}|\tilde{u}',\mspace{-2mu}q') \nn\\ &  \quad \times  \mathbbm{1}\{x=f(\tilde{u},s')\} P_{S^+|X,S,Q}(s|x,s',q') \mathbbm{1}\{q=g(q',x)\} \Big|\nn\\
    &= \frac{1}{n} \Big| \sum\nolimits_{i=1}^{n} \big[ P_{S_{i-1},\tilde{U}_{i-1},Q_{i-1}}(s,\tilde{u},q) \nn\\
    & \;-\mspace{-4mu}\sum\nolimits_{s',\tilde{u}',q',x} \mspace{-5mu}
    P_{S_{i-1},\tilde{U}_{i-1},Q_{i-1}}(s',\mspace{-3mu}\tilde{u}',\mspace{-3mu}q')
    P_{\tilde{U}_i|\tilde{U}_{i-1},Q_{i-1}}\mspace{-3mu}(\tilde{u}|\tilde{u}',\mspace{-3mu}q')\nn\\ &  \quad \times\mathbbm{1}\{x=f(\tilde{u},s')\} P(s|x,s',q') \mathbbm{1}\{q=g(q',x)\} \big] \Big|\nn\\         
    &= \frac{1}{n} \Big| \sum\nolimits_{i=2}^{n} \big[ P_{S_{i-1},\tilde{U}_{i-1},Q_{i-1}}(s,\tilde{u},q)     
    \nn\\ &  \;
    -\mspace{-4mu} \sum\nolimits_{s',\tilde{u}',q',x}  
    \mspace{-5mu} P_{S_{i-1},\tilde{U}_{i-1},Q_{i-1}}(s',\mspace{-3mu}\tilde{u}',\mspace{-3mu}q')P_{\tilde{U}_i|\tilde{U}_{i-1},Q_{i-1}}\mspace{-2mu}(\tilde{u}|\tilde{u}',\mspace{-3mu}q') \nn\\ &  \quad \times \mathbbm{1}\{x=f(\tilde{u},s')\} P(s|x,s',q') \mathbbm{1}\{q=g(q',x)\} 
    \big] \nn\\
    &  \quad + P_{S_{0},\tilde{U}_{0},Q_{0}}(s,\tilde{u},q) - P_{S_{n},\tilde{U}_{n},Q_{n}}(s,\tilde{u},q)\Big| \le \frac{1}{n}, 
\end{align}
where the inequality step follows since for any $i \in [1:n]$:
\begin{align}
& \sum\nolimits_{s',\tilde{u}',q',x}         
    \mspace{-5mu} P_{S_{i-1},\tilde{U}_{i-1},Q_{i-1}}(s',\tilde{u}',q')P_{\tilde{U}_i|\tilde{U}_{i-1},Q_{i-1}}\mspace{-3mu}(\tilde{u}|\tilde{u}',q')
    \nn\\ & \times \mspace{-5mu} \mathbbm{1}\mspace{-3mu}\{x\mspace{-4mu}=\mspace{-5mu}f(\tilde{u},s')\} \mspace{-2mu} P(\mspace{-2mu}s|x,s',q') \mspace{-2mu}\mathbbm{1}\mspace{-2mu}\{q\mspace{-5mu}=\mspace{-5mu}g(\mspace{-2mu}q',\mspace{-3mu}x\mspace{-2mu})\} \mspace{-5mu}=\mspace{-5mu} P_{S_{i},\tilde{U}_{i},Q_{i}}\mspace{-3mu}(\mspace{-2mu}s,\mspace{-3mu}\tilde{u},\mspace{-3mu}q\mspace{-2mu}). \nn
\end{align}
This completes the first part of the proof. 

In the second part of the proof, we derive from \eqref{eq:step_EH}:
\begin{align}
    &\bar{C}_{\text{BEHC}}(N) \nn\\
    &\le \lim_{n \to \infty} \max_{{\{P(\tilde{u}_i|\tilde{u}_{i-1},q_{i-1})\}}_{i=1}^n} \mspace{-4mu}
 \frac{1}{n}  \mspace{-4mu}\sum\nolimits_{i=1}^n \mspace{-4mu}I(\tilde{U}^{(N)}_i\mspace{-3mu},\tilde{U}^{(N)}_{i-1};X_i|Q_{i-1}^{(N)})
    \nn\\
     & \le\lim_{n\to \infty}
    \max_{P \in \mathcal{D}_{\frac{1}{n}}} I(\tilde{U}^{+(N)},\tilde{U}^{(N)};X|Q^{(N)}) \nn\\
    & \stackrel{(a)}= \max_{P \in \mathcal{D}_{0}} I(\tilde{U}^{+(N)},\tilde{U}^{(N)};X|Q^{(N)}) \nn\\
    & \stackrel{(b)} = \max_{P(\tilde{u}^+|\tilde{u},q)}  I(\tilde{U}^{+(N)},\tilde{U}^{(N)};X|Q^{(N)}). \label{eq:step_stationary_EH}
\end{align}
Step (a) follows because any $P\in \mathcal{D}_0$ satisfies $P\in \cap_{n=1}^{\infty} \mathcal{D}_{\frac{1}{n}}$, as $\frac{1}{n}>0, \; \forall n\in \mathbb{N}$. 
Conversely, any $P\in \cap_{n=1}^{\infty} \mathcal{D}_{\frac{1}{n}}$ satisfies $P\in \mathcal{D}_0$
because $\frac{1}{n}$ monotonically decreases as $n\to \infty$. Thus,
$\lim_{n\to \infty} \mathcal{D}_{\frac{1}{n}} = \mathcal{D}_{0}$.
Step (b) follows from the fact that $\mathcal{D}_{0}$ represents the set of all distributions $P(\tilde{u},q,u^+)P(s|\tilde{u},q)$ inducing a stationary distribution on $(S,\tilde{U},Q)$, i.e., 
\begin{align}
    \mspace{-5mu} P_{S,\tilde{U},Q}\mspace{-2mu}(\mspace{-3mu}s,\mspace{-4mu}\tilde{u},\mspace{-4mu}q\mspace{-2mu})\mspace{-4mu}=\mspace{-4mu}P_{S^+,\tilde{U}^+\mspace{-4mu},Q^+}\mspace{-2mu}(\mspace{-3mu}s,\mspace{-4mu}\tilde{u},\mspace{-4mu}q\mspace{-2mu}), \forall s,\mspace{-4mu} \tilde{u},\mspace{-4mu} q\mspace{-4mu} \in \mspace{-4mu} \cS \mspace{-4mu} \times \mspace{-4mu} \tilde{\cU}^{(Q=q)} \mspace{-6mu} \times \mspace{-6mu} \cQ. \label{eq:stationarity_step_proof}
\end{align}
Recall that since $\cS, \cQ^{(N)}$ are finite, and $\tilde{\cU}$ will be shown to be finite in the next theorem, 
there always exists a stationary distribution $\pi(s,\tilde{u},q)$ (not necessarily unique) for any $P(\tilde{u}^+|\tilde{u},q)$. Thus, $\mathcal{D}_{0}$ is non-empty.

Finally, the proof concludes by renaming $\tilde{U}$ in place of $U$ and demonstrating that the maximization in Step (b) can be restricted to $\mathcal{P}_\pi$, the set of $P(u^+|u,q)$ distributions that induce a unique stationary distribution. This conclusion follows from the connectivity (Definition~\ref{def:connectivity}) of the boosted BEHC($N$), and can be proven using the approach outlined in \cite[Lem.~1]{NOST}. \qed

We conclude this section with the proof of Step 3.

\subsubsection{Proof of Theorem~\ref{thr:cardinality} (Step 3)} 
We establish the cardinality bounds on $|\tilde{\cU}|$ and $|\tilde{\cU}^+|$, which are identical by the stationary distribution $\pi(s,\tilde{u},q)$ \eqref{eq:stationarity_step_proof}. We also validate the optimality of the policy structure by analyzing the transformations of the auxiliary RVs $U \Rightarrow \hat{U}\Rightarrow \tilde{U}$ throughout the proof of Theorem~\ref{thr:FSC_Qgraph_UB}. 
Initially, the set of strategies satisfies $|\cU|\le|\cX|^{|{\cS}| |\cQ^{(N)}|}$, where $|\cQ^{(N)}|\mspace{-3mu}=\mspace{-3mu}N\mspace{-2mu}+\mspace{-2mu}1$. In \eqref{eq:joint_boosted}, Step~$(*)$ introduces $\hat{u}_i\triangleq(u_i,q_{i-1})$, mapping $x_i=f_{\hat{u}_i}(s_{i-1})$, which reduces $|\hat{\cU}|$ to $|\cX|^{|\cS|}$. 
However, due to the battery constraint $X_i\mspace{-3mu}=\mspace{-3mu}0$ when $S_{i-1}\mspace{-3mu}=\mspace{-3mu}0$, the cardinality of this auxiliary RV further reduces to $|\hat{\cU}_i|\mspace{-3mu}=\mspace{-3mu} |\cX|^{|\cS|-1}\mspace{-3mu}=\mspace{-3mu}2$, as there are only two feasible strategies out of four in the set $\{f_{\hat{u}} (s):\cS \to \cX \}$, as shown in Table~\ref{table:strategies_EH}. Specifically, strategy $\hat{u}^+\mspace{-3mu}=\mspace{-3mu}a$ corresponds to always transmitting $X\mspace{-3mu}=\mspace{-3mu}0$, regardless of the battery state, while strategy $\hat{u}^+\mspace{-3mu}=\mspace{-3mu}b$ represents an attempt to transmit $X\mspace{-3mu}=\mspace{-3mu}1$, which succeeds only if the battery is charged.



\begin{table}[t]
\centering 
\begin{tabular}[b]{||l|c|c||} 
\hline \hline
$x=f_{\hat{u}^+}(s)$ & $s=0$ & $s=1$ \\
\hline \hline
$\hat{u}^+=a$ (feasible -- always transmit $'0'$) & $0$ & $0$ \\
\hline
$\hat{u}^+=b$ (feasible -- attempt to transmit $'1'$) & $0$ & $1$ \\
\hline
$\hat{u}^+=c$ (infeasible) & $1$ & $0$ \\
\hline
$\hat{u}^+=d$ (infeasible) & $1$ & $1$ \\
\hline \hline
\end{tabular}
\caption{Feasibility of strategies under the battery constraint.}
\label{table:strategies_EH}
\end{table}

Now, we revisit Step (c) in \eqref{eq:step_EH}, which states
\begin{align}
    I(\hat{U}_{K_{i-1}}^i; X_i | Q_{i-1})    
    \mspace{-3mu}=\mspace{-3mu}I(\tilde{U}_i,\tilde{U}_{i-1};X_i|Q_{i-1}), \,i\in[1\mspace{-3mu}:\mspace{-3mu}n], \label{eq:StepC_revisit}
\end{align}
where $\tilde{U}_{i-1} \mspace{-3mu} \triangleq \mspace{-3mu}\hat{U}_{K_{i-1}}^{i-1}$, and redefine $\tilde{U}_{i-1}$ to establish the auxiliary sets in 
\eqref{eq:Uset_initialNodesProof}–\eqref{eq:Uset_lastNodeProofUB}.
We analyze each of the two cases corresponding to whether $Q_{i-1}\mspace{-3mu}\in\mspace{-3mu} [0\mspace{-3mu}:\mspace{-3mu}N\mspace{-3mu}-\mspace{-3mu}1]$ or $Q_{i-1}\mspace{-3mu}=\mspace{-3mu}N$.

For Case $(i)$ ($X_{{K}_{i-1}}^{i-1}=(1,\mathbf{0}^{i-1-{K_{i-1}}})$), corresponding to nodes $Q_{i-1}\mspace{-3mu}\in\mspace{-3mu} [0\mspace{-3mu}:\mspace{-3mu}N\mspace{-3mu}-\mspace{-3mu}1]$, we employ the following lemma.

\begin{lemma}
\label{lem:tryTransmit1}
For the boosted BEHC($N$) at time $i$, given $x^{i-1}$ such that $x_{i-N}^{i-1}\mspace{-3mu}\neq \mspace{-3mu}\mathbf{0}^N$, and given $\hat{u}^i$ containing a component $\hat{u}_l\mspace{-4mu}=\mspace{-4mu}b$ for some time index $l\in [k_{i-1}(x_{i-N}^{i-1}):i\mspace{-3mu}-\mspace{-3mu}1]$, 
the following holds: 
$P(x_i|\hat{u}^i,x^{i-1})=P(x_i|\hat{u}_l^i,q_{i-1}(x_{i-N}^{i-1}))$.
\end{lemma}

The proof of Lemma~\ref{lem:tryTransmit1} appears in Appendix~\ref{Appendix:prf_lem_tryTransmit1}. From this lemma, we deduce that for any time $i$, node $q_{i-1} \in [0:N-1]$ and vector $\hat{u}^i$, we only need to consider the maximal time index $l\in [k\mspace{-3mu}:\mspace{-3mu}i\mspace{-3mu}-\mspace{-3mu}1]$ for which $\hat{u}_l=b$, while all previous strategies $\hat{u}^{l-1}$ can be ignored as they do not affect $I(\hat{U}_{K_{i-1}}^i; X_i | Q_{i-1}=q_{i-1})$.
Thus, defining $m \mspace{-3mu}\triangleq \mspace{-3mu} \max \{l \mspace{-3mu} \in  \mspace{-3mu} [k_{i-1}\mspace{-3mu} :\mspace{-3mu} i\mspace{-3mu}-\mspace{-3mu}1] \;| \; \hat{u}_{l}\mspace{-3mu}=\mspace{-3mu} b \}$, we obtain $\hat{u}_m^{i-1}\mspace{-3mu}=\mspace{-3mu}(b,\mspace{-2mu}a,\mspace{-2mu}a,\mspace{-2mu}\dots,\mspace{-2mu}a)$. Due to the support of $m$, there are $i\mspace{-3mu}-\mspace{-3mu}k_{i-1}=q_{i-1}+1$ possible combinations for $\hat{u}_m^{i-1}$. Hence, we establish a one-to-one mapping 
to $[0\mspace{-3mu}:\mspace{-3mu}q_{i-1}]$ via $\tilde{u}_{i-1}\mspace{-3mu}\triangleq \mspace{-3mu} i\mspace{-3mu} -\mspace{-3mu} 1\mspace{-3mu} -\mspace{-3mu} m$. We then define $\tilde{u}_i=\mathbbm{1}\{\hat{u}_i=a\} (\tilde{u}_{i-1}\mspace{-3mu}+\mspace{-3mu}1)$, and map $X_i$ via
\begin{align}
     \tilde{f}(\tilde{U}_i,S_{i-1})&=S_{i-1} \mathbbm{1}\{\tilde{U}_i^+=0\}.
    \label{eq:f_tilde}
\end{align}
Since there is a one-to-one mapping between $\hat{U}_{K_{i-1}}^i$ and $(\tilde{U}_i,\tilde{U}_{i-1})$ given $q_{i-1}$, we obtain $I(\hat{U}_{K_{i-1}}^i; X_i | Q_{i-1}=q_{i-1})=I(\tilde{U}_i,\tilde{U}_{i-1}; X_i | Q_{i-1}=q_{i-1})$.
This establishes $\tilde{\cU}=[0:N]$ and \eqref{eq:policy_first_nodes_proof}, and yields the auxiliary sets in \eqref{eq:Uset_initialNodesProof}.

For Case $(ii)$ ($X_{i-N}^{i-1}\mspace{-3mu}=\mspace{-3mu}\mathbf{0}^N$), corresponding to $Q_{i-1}\mspace{-3mu}=\mspace{-3mu}N$, we obtain
$P\mspace{-1mu}(x_i|\hat{u}^i\mspace{-3mu},\mspace{-3mu}x^{i-N-1}\mspace{-3mu},\mspace{-3mu}x_{i-N}^{i-1}\mspace{-4mu}=\mspace{-4mu}\mathbf{0}^N)\mspace{-4mu}=\mspace{-4mu}P\mspace{-1mu}(x_i|\hat{u}_i\mspace{-1mu},\mspace{-2mu}q_{i-1}\mspace{-2mu}=\mspace{-2mu}N)$ (see derivation in \eqref{eq:opt2}, Appendix~\ref{appendix:lemma_boosted}). Thus, $I(\hat{U}_{K_{i-1}}^i; X_i | Q_{i-1}=N)=\mspace{-3mu}I(\hat{U}_{i};X_i|Q_{i-1}\mspace{-3mu}=\mspace{-3mu}N)$, as only $\hat{U}_i$ is relevant. Consequently, we redefine $\tilde{u}_i=\hat{u}_i$, relabeling $(a,b) \mspace{-3mu} \Leftrightarrow \mspace{-3mu}(0,1)$, without requiring $\tilde{u}_{i-1}$. That is, we set $|\tilde{\cU}^{(Q=N)+}\mspace{-2mu}|\mspace{-6mu}\triangleq\mspace{-6mu} 2$ and $|\tilde{\cU}^{(Q=N)}\mspace{-2mu}|\mspace{-6mu}=\mspace{-6mu}0$.
However, we set $\tilde{\cU}^{(Q=N)}\mspace{-3mu}=\mspace{-3mu}[0\mspace{-3mu}:\mspace{-3mu}N]$ to align with $\tilde{\cU}^{+(Q=N-1)}\mspace{-3mu}=\mspace{-3mu}[0\mspace{-3mu}:\mspace{-3mu}N]$, since node $Q^+\mspace{-3mu}=\mspace{-3mu}N$ is reached by $g(Q\mspace{-3mu}=\mspace{-3mu}N\mspace{-3mu}-\mspace{-3mu}1,X\mspace{-3mu}=\mspace{-3mu}0)$. Additionally, stationarity on $(S,\tilde{U},Q)$ must be preserved. Hence, \eqref{eq:policy_last_node_UB_proof} and \eqref{eq:Uset_lastNodeProofUB} are established.

Considering the redefinitions of $\tilde{U}_i,\tilde{U}_{i-1}$ in both cases, along with mapping $X_i$ via \eqref{eq:f_tilde}, we conclude that \eqref{eq:StepC_revisit} holds. The resulting joint distribution $P(s_0^n,\tilde{u}^n,x^n,q^n)$ is given by 
\begin{align}
&P_{S_0,Q_0}(s_0,q_0) \prod\nolimits_{i=1}^n P(\tilde{u}_i|\tilde{u}_{i-1},q_{i-1}) \mspace{-2mu}\mathbbm{1}\{x_i = \tilde{f}(\tilde{u}_i,s_{i-1})\nn\\ &\times P_{S^+|X,S,Q}(s_i|x_i,s_{i-1},q_{i-1})  \mathbbm{1}\{q_i=g(q_{i-1},x_i)\}
    \text{.} \label{eq:joint_boostedProofCard}
\end{align}
Thus, the policy structure and auxiliary RV sets are established, and \eqref{eq:step_EH}, Step (c), follows with \eqref{eq:joint_boostedProofCard}.

To maintain consistency with the theorem's notation, we rename $\tilde{U}_i, \tilde{f}$ by $U_i, f$, respectively. It remains to show that the policy structure induces the marginal $\pi_{S|U,Q}$ in \eqref{eq:s_given_uq_upper}. The proof follows by induction, as in Lemma~\ref{lem:marginal_s_given_uq}, except for the last node $Q=N$, because for any $Q\in[0:N-1]$, the conditionals $P_{S^+|,X,S,Q}$ and $P_{S^+|X,S}$ are equal by \eqref{eq:modifiedChannelLaw}. 

Node $Q^+=N$ is reached either by $g(Q\mspace{-3mu}=\mspace{-3mu}N-1,X\mspace{-3mu}=\mspace{-3mu}0)$ or by $g(Q\mspace{-3mu}=\mspace{-3mu}N,X\mspace{-3mu}=\mspace{-3mu}0)$. Hence, for any $u$:
\begin{align}
    &P_{S|U,Q}(0|u,N)\stackrel{(a)}=P_{S^+|U^+,Q^+}(0|u,N)\nn\\
    &\stackrel{(b)}=P_{S^+|U^+,Q^+,X}(0|u,N,0) \nn\\
    &=\sum\nolimits_{q\in\{N-1,N\},s} P_{Q,S,S^+|U^+,Q^+,X}(q,s,0|u,N,0)\nn\\
    &\stackrel{(c)}=\mspace{-22mu}\sum_{q\in\{N-1,N\},s} \mspace{-27mu}P_{Q,S|U^+,Q^+,X}(q,s|u,N,\mspace{-3mu}0)P_{S^+|X,S,Q}(0|0,\mspace{-3mu}s,\mspace{-3mu}q)\mspace{-3mu}\stackrel{(d)}= 0. \nn 
\end{align}
Step (a) follows from the stationary distribution. Step (b) follows because $Q^+=N$ is reached only by $X=0$, as it has no incoming edge $'X=1'$. Step (c) follows from the Markovity of $P_{S^+|X,S,Q}$ implied by the joint distribution \eqref{eq:modifiedChannelLaw}. Step (d) follows from 
\eqref{eq:modifiedChannelLaw}.
This yields $P_{S|U,Q}(0|u,N)=0$, completing the proof.
\qed

\begin{remark}
    The proof of Theorem~\ref{thr:cardinality} provides an intuitive interpretation of the auxiliary RV $U_{i-1}$: it represents the time elapsed since the encoder last \textit{attempted} to transmit a $'1'$, with success depending on whether the battery was charged. The policy $P_{U_i|U_{i-1},Q_{i-1}}(u^+|u,q), u^+\in \{0,u+1\}$ describes the likelihood of attempting another $'1'$ ($u^+=0$) versus deliberately transmitting a $'0'$ ($u^+=u+1$). This aligns with $\pi(S=0|u,q)$ in \eqref{eq:s_given_uq_upper}: if the last attempt was $u$ steps ago, the probability of an empty battery is $\bar{\eta}^{u+1}$, unless all $N$ previous transmission were zeros, ensuring a charged battery.
\end{remark}

We now proceed to prove the convergence of both the upper and lower bound sequences to $C_{\text{BEHC}}$.


\subsection{Proof of Bounds Convergence to 
\label{subsec:ProofThrConvergence}
$C_{\text{BEHC}}$}
To complete the proofs of Theorems~\ref{thr:EH_cvx_lower} and~\ref{thr:EH_cvx_upper}, we introduce the following theorem.

\begin{theorem}
\label{thr:convergence}
The sequences of solutions to the convex optimization problems for the lower and upper bounds converge to $C_{\text{BEHC}}$ as $N$ tends to infinity.
\end{theorem}

Let $a_N$ and $b_N$ denote the solutions of the convex optimization problems for the lower and upper bounds in \eqref{optProb_EH_LB} and \eqref{optProb_EH_UB}, respectively. Thus, we aim to prove that 
\begin{align}
\label{eq:thrConvergence}
    \lim_{N\to \infty} a_N= C_{\text{BEHC}}=\lim_{N\to \infty} b_N.
\end{align}

\begin{proof}
We divide the proof into two parts. In the first part, we prove that both sequences converge to the same limit denoted by $J$, i.e.,
\begin{align}
    \lim_{N\to \infty} a_N = \lim_{N\to \infty} b_N=J. \label{eq:conv_part1}
\end{align}
In the second part of the proof, for any $N\ge0$, we have 
$a_N \le C_{\text{BEHC}} \le \bar{C}_{\text{BEHC}}(N) \le b_N$,
and define the constant sequence $c_N\triangleq C_{\text{BEHC}}$ to obtain $a_N \le c_N \le b_N$.
Taking the limit as $N\to \infty$, and using \eqref{eq:conv_part1} alongside the squeeze theorem, we find $\lim_{N\to\infty} c_N=J$. Since $c_N$ is constant, it follows that 
$\lim_{N\to\infty} c_N=C_{\text{BEHC}}$, which implies $J=C_{\text{BEHC}}$, thus concluding the proof. Therefore, it remains to prove \eqref{eq:conv_part1}.

For the sequence of solutions to the convex optimization problems \eqref{optProb_EH_UB} for the lower bound, we denote $\underline{I}_{\underline{P}^*}(U^{+(N)},U^{(N)};X|Q^{(N)})= a_N$, representing the conditional mutual information induced by an optimal joint distribution $\underline{P}^*_{S,U,Q,X,S^+,U^+,Q^+}$.
Similarly, for the sequence of solutions to the convex optimization problems \eqref{optProb_EH_UB} for the upper bound, we denote $\bar{I}_{\bar{P}^*}(U^{+(N)},U^{(N)};X|Q^{(N)})=b_N$,
representing the conditional mutual information induced by an optimal joint distribution $\bar{P}^*_{S,U,Q,X,S^+,U^+,Q^+}$. We construct a joint distribution $\tilde{P}_{S,U,Q^,X,S^+,U^+,Q^+}$
that satisfies all the constraints for a valid lower bound in Theorem~\ref{thr:EH_cvx_lower}.
$\tilde{P}$ induces the lower bound $\underline{I}_{\tilde{P}}(U^{+(N)},U^{(N)};X|Q^{(N)})$, which represents the conditional mutual information based on the $N+1$ sized $Q$-graph in Fig.~\ref{fig:EH_LB_Graph}. The construction mirrors the policy of $\bar{P}^*$, which corresponds to the upper bound, for any node $q\in [0:N-1]$, i.e.,
\begin{align}
    \label{eq:mimicPolicy}
    \mspace{-9mu} \tilde{P}_{U^+|U,Q}(\mspace{-2mu}u^+\mspace{-4mu}|u,\mspace{-4mu}q)\mspace{-4mu}=\mspace{-4mu} \bar{P}^*_{U^+|U,Q}(u^+\mspace{-4mu}|u,\mspace{-4mu}q),  q\mspace{-4mu}\in\mspace{-4mu} [0\mspace{-4mu}:\mspace{-4mu}N\mspace{-4mu}-\mspace{-4mu}1],\mspace{-4mu} \,\forall u,\mspace{-4mu} u^+\mspace{-4mu}, 
\end{align}
while for the last node, $Q=N$, it satisfies Constraint \eqref{eq:policy_last_node}, ensuring a valid joint distribution for the lower bound. This construction induces that $\tilde{P}(x,u,u^+|q)=\bar{P}^*(x,u,u^+|q)$ for all $q \in [0:N-1]$ because
\begin{align}
\label{eq:joint_mimicUB}
    & \tilde{P}(x,u,u^+|q)=\sum\nolimits_{s} \tilde{P}(s,x,u,u^+|q) \nn\\
    &=\sum\nolimits_{s} \tilde{P}(u|q) \pi_{S|U,Q}(s|u,q)\tilde{P}(u^+|u,q) \mathbbm{1}\{x = f(u^+,s)\}     \nn\\
    &\stackrel{(*)}=\sum\nolimits_{s} \bar{P}^{*}(u|q) \pi_{S|U,Q}(s|u,q)\bar{P}^{*}(u^+|u,q) \mathbbm{1}\{x = f(u^+,s)\}     \nn\\
    & =\bar{P}^*(x,u,u^+|q), 
\end{align}
where Step $(*)$ follows from \eqref{eq:mimicPolicy}, the fact that both distributions share the same function $f(U^+,S)$ as given in \eqref{eq:x_behc}, and the same conditional stationary distribution $\pi_{S|U,Q}(s|u,q)$ for all nodes $q\ne N$ (see \eqref{eq:s_given_uq_lb} vs. \eqref{eq:s_given_uq_upper}). Furthermore, by induction, it holds that
\begin{align}
\label{eq:inductionClaimUgivenQ}
\tilde{P}(u|q)=\bar{P}^{*}(u|q), \forall q \in [0:N-1]    
\end{align}
\textit{Base case:} For $q=0$. we have $\tilde{P}_{U|Q}(0|0)=1=\bar{P}^{*}_{U|Q}(0|0)$ because
$\pi_{U,S|Q}(0,0|0)=\bar{\eta}, \quad \pi_{U,S|Q}(0,1|0)=\eta$, as shown in the proof of Lemma~\ref{lemma:eh_lb_BCJR}. This results applies to any policy of the lower bound in Theorem~\ref{thr:EH_cvx_lower} and, in particular, holds for both $\tilde{P}$ and $\bar{P}^{*}$.\\ \textit{Inductive Hypothesis:
}: Assume that 
\begin{align}
\label{eq:inductionHypothesisMimic}
    \tilde{P}(u|q)=\bar{P}^{*}(u|q), q\in [0:N-2], \quad \forall u\in [0:q].
\end{align}
\textit{Inductive Step}: We need to prove that
\begin{align}
\tilde{P}_{U|Q}(u|q+1)=\bar{P}^{*}_{U|Q}(u|q+1), \quad \forall u \in [0:q+1]. \label{eq:inductionStep}
\end{align}
From the stationarity on $(S,U,Q)$, which holds for both $\tilde{P}$ and $\bar{P}^{*}$, it follows that $\pi_{U^+|Q^+} (u|q)=\pi_{U|Q}(u|q), \forall u,q$. Thus, proving \eqref{eq:inductionStep} reduces to showing
\begin{align}
\mspace{-9mu} \tilde{P}_{U^+|Q^+}\mspace{-3mu}(u^+|q\mspace{-3mu}+\mspace{-3mu}1)\mspace{-4mu}=\mspace{-4mu}\bar{P}^{*}_{U^+|Q^+}\mspace{-3mu}(u^+|q\mspace{-3mu}+\mspace{-3mu}1), \forall u^+ \mspace{-4mu}\in\mspace{-4mu} [0:q\mspace{-3mu}+\mspace{-3mu}1]. 
\label{eq:inductionAlternativeStep}
\end{align}
For both $Q$-graphs of lower and upper bounds in Fig.~\ref{fig:EH_LB_Graph} and Fig.~\ref{fig:EH_UB_Graph}, node $Q^+=q+1 \in [1:N-1]$ is reached solely by node $Q=q$ with input $X=0$, i.e. $g(q,0)$. Hence, the joint distribution 
$\tilde{P}_{U^+,S^+|Q^+}(u^+,s^+|q+1)$ is given by
\begin{align} 
&\frac{\sum_{u,s}
\mspace{-15mu} \tilde{P}(u|q)\pi(s|u,\mspace{-3mu} q)
\tilde{P}(u^+\mspace{-3mu} |u, \mspace{-3mu} q) \mspace{-3mu} \mathbbm{1}\mspace{-3mu} \{0 \mspace{-3mu} = \mspace{-5mu} f(u^+\mspace{-3mu} ,\mspace{-3mu}s)\} \mspace{-3mu} P_{S^+|X,S}\mspace{-3mu}(s^+\mspace{-3mu} |0, \mspace{-3mu} s)}{\sum_{u,u'^+,s} \tilde{P}(u|q)\pi(s|u, \mspace{-3mu} q)\tilde{P}(u'^+|u, \mspace{-3mu} q)\mspace{-3mu} \mathbbm{1}\mspace{-3mu} \{0 = f(u'^+, \mspace{-3mu} s)\}}, \nn
\end{align}
which equals the following expression due to the induction hypothesis \eqref{eq:inductionHypothesisMimic}, Eq.~\eqref{eq:mimicPolicy}, and the fact that both $\tilde{P}$ and $\bar{P}^*$ share the same $f(U^+,S)$, $\pi_{S|U,Q}$ and $P_{S^{+}|X,S}$ for all $Q\ne N$: 
\begin{align}
&\frac{\sum_{u,s}\mspace{-16mu} \bar{P}^*(u|q)\pi(\mspace{-3mu} s|u, \mspace{-3mu} q)\mspace{-3mu} \bar{P}^*\mspace{-3mu} (u^+\mspace{-3mu} |u, \mspace{-3mu}q) \mathbbm{1}\{0 \mspace{-3mu}=\mspace{-3mu} f(u^+\mspace{-5mu} ,\mspace{-3mu}s)\} P_{S^+|X,S}(s^+\mspace{-3mu}|0,\mspace{-3mu}s)}{\sum_{u,u'^+,s} \bar{P}^*(u|q)\pi(s|u,\mspace{-3mu} q) \bar{P}^*(u'^+\mspace{-3mu} |u, \mspace{-3mu} q)\mathbbm{1}\{0 \mspace{-3mu} =\mspace{-3mu}  f(u'^+,\mspace{-3mu} s)\}} \nn\\
&=\bar{P}^*_{U^+,S^+|Q^+}(u^+,s^+|q+1). \nn
\end{align}
Consequently, the marginals are also equal, i.e., \eqref{eq:inductionAlternativeStep} holds. Thus, \eqref{eq:inductionClaimUgivenQ} holds as well, and \eqref{eq:joint_mimicUB} is established.

Now, we evaluate the difference between the bounds
\begin{align}
    &0 \le b_N-a_N \nn\\
    & =\mspace{-3mu}\bar{I}_{\bar{P}^*}(U^{+(N)}\mspace{-3mu},U^{(N)};X|Q^{(N)})\mspace{-3mu}-\mspace{-3mu}\underline{I}_{\underline{P}^*}(U^{+(N)}\mspace{-3mu},U^{(N)};X|Q^{(N)})  \nn\\
    & \stackrel{(a)}\le \mspace{-3mu} \bar{I}_{\bar{P}^*}(U^{+(N)}\mspace{-3mu},U^{(N)};X|Q^{(N)})\mspace{-3mu}-\mspace{-3mu}\underline{I}_{\tilde{P}}(U^{+(N)}\mspace{-3mu},U^{(N)};X|Q^{(N)})  \nn\\
    & =\mspace{-8mu}\sum_{q=0}^{N-1} \mspace{-8mu} \bar{P}^*_{Q}(q) \bar{I}_{\bar{P}^*}\mspace{-3mu} (U^+\mspace{-5mu} ,\mspace{-3mu} U; \mspace{-3mu} X | Q\mspace{-3mu} = \mspace{-3mu} q) \mspace{-4mu} + \mspace{-4mu} \bar{P}^*_{Q}(\mspace{-2mu} N \mspace{-2mu}) \bar{I}_{\bar{P}^*}(U^+\mspace{-5mu} ,U ; \mspace{-2mu}X|Q\mspace{-4mu} = \mspace{-4mu}N\mspace{-2mu}) \nn\\ & - \mspace{-9mu} \sum_{q=0}^{N-1} \mspace{-7mu}\tilde{P}_{Q}(q)\mspace{-1mu} \underline{I}_{\tilde{P}}(U^+\mspace{-3mu},U;\mspace{-2mu} X |Q\mspace{-3mu}=\mspace{-3mu}q) \mspace{-5mu}-\mspace{-5mu} \tilde{P}_{Q}(\mspace{-2mu}N\mspace{-2mu}) \underline{I}_{\tilde{P}}(U^+\mspace{-3mu},\mspace{-2mu} U;\mspace{-2mu} X|Q\mspace{-3mu}=\mspace{-3mu}N) \nn\\
    & \stackrel{(b)}=\sum\nolimits_{q=0}^{N-1} \left(\bar{P}^*_{Q}(q) -\tilde{P}_{Q}(q)\right) \bar{I}_{\tilde{P}}(U^+,U;X|Q=q)+ \nn\\
    & \quad \mspace{-6mu} \bar{P}^*_{Q}(N) \bar{I}_{\bar{P}^*}(U^+\mspace{-3mu} ,\mspace{-3mu}U;X|Q\mspace{-3mu}=\mspace{-3mu} N) \mspace{-4mu}-\mspace{-4mu} \tilde{P}_{Q}(N) \underline{I}_{\tilde{P}}(U^+\mspace{-3mu},\mspace{-3mu}U;X|Q\mspace{-3mu}=\mspace{-3mu} N) \nn\\
    &\stackrel{(c)}\le \mspace{-3mu} \bar{P}^*_{Q}(N) \bar{I}_{\bar{P}^*}(U^+,U;X|Q\mspace{-3mu}=\mspace{-3mu}N) \mspace{-5mu}\stackrel{(d)}=\mspace{-5mu} \bar{P}^*_{Q}(N) H_2(\bar{P}^*_{U^+|Q}(0|N)) \nn\\
    & \stackrel{(e)}\le \frac{H_2(\bar{P}^*_{U^+|Q}(0|N))}{N \bar{P}^*_{U^+|Q}(0|N) +1} \le \max_{p\in[0,1]} \frac{H_2(p)}{N p +1}.    \label{eq:gap}
\end{align}
Step~(a) follows since $\tilde{P}$ satisfies the conditions for the lower bound optimization problem, but may not be optimal.
Step~(b) utilizes \eqref{eq:joint_mimicUB} 
to evaluate $\underline{I}_{\tilde{P}}(U^+,U;X|Q\mspace{-3mu}=\mspace{-3mu}q)\mspace{-3mu}=\mspace{-3mu}\bar{I}_{\bar{P}^*}(U^+,U;X|Q\mspace{-3mu}=\mspace{-3mu}q)$ for $q \in [0:N-1]$. 
Step~(c) relies on demonstrating that 
\begin{align}
\label{eq:statLastNode_mimic}
    \bar{P}^*_{Q}(q) \le \tilde{P}_{Q}(q), \quad \forall q\in [0:N-1].
\end{align}
To establish this, we observe that the transition probabilities $P_{Q^+|Q}(q^+|q)$ from node $q\in [0:N-1]$ to node $q^+\in [0:N]$ are identical for $\bar{P}^*$ and $\tilde{P}$. This follows from $\tilde{P}(x|q)\mspace{-3mu}=\mspace{-3mu}\bar{P}^*(x|q)$, as implied by \eqref{eq:joint_mimicUB}, and the fact that the $Q$-graphs for both the lower and upper bounds are identical, except at node $Q\mspace{-3mu}=\mspace{-3mu}N$. At node $Q\mspace{-3mu}=\mspace{-3mu}N$, however, there is a key distinction: the self-loop in $\bar{P}^*$ is absent in $\tilde{P}$, i.e.,
$\tilde{P}_{Q^+|Q}(N|N)=0\le \bar{P}^*_{Q^+|Q}(N|N)$. In contrast, $\tilde{P}_{Q^+|Q}(0|N)\ge \bar{P}^*_{Q^+|Q}(0|N)$. As a result, the corresponding stationary distributions satisfy 
$\bar{P}^*_{Q}(N) \ge \tilde{P}_{Q}(N)$ and the necessary Inequality \eqref{eq:statLastNode_mimic}.
Step~(d) follows from $\bar{I}_{\bar{P}^*}(U^+,U;X|Q\mspace{-3mu}=\mspace{-3mu}N)\mspace{-3mu}=\mspace{-3mu}H_{\bar{P}^*}(X|Q\mspace{-3mu}=\mspace{-3mu}N)-H_{\bar{P}^*}(X|U^+,U,Q\mspace{-3mu}=\mspace{-3mu}N)$ and simplifying each entropy. Specifically, using \eqref{eq:x_behc} and \eqref{eq:s_given_uq_upper}, the latter entropy simplifies to
$\mspace{-3mu}H_{\bar{P}^*}(X|S\mspace{-3mu}=\mspace{-3mu}1,U^+\mspace{-3mu}=\mspace{-3mu}0,U,Q\mspace{-3mu}=\mspace{-3mu}N)\mspace{-3mu}=\mspace{-3mu}0$.
Step~(e) follows from the fact that $\bar{P}^*_{Q}(N) \le \bar{P}'_{Q}(N)$, where $\bar{P}'$ is a joint distribution distributed as \eqref{eq:Qgraph_joint_distUB} with the upper bound $Q$-graph in Fig.~\ref{fig:EH_UB_Graph}, similar to $\bar{P}^*$, but constructed to maximizes the $Q$-graph stationary distribution at node $Q\mspace{-3mu}=\mspace{-3mu}N$ compared to $\bar{P}^*$. The distribution $\bar{P}'$ is constructed by setting $\bar{P}'_{U^+|U,Q}(0|u,q)\mspace{-3mu}=\mspace{-3mu}0$ for any $q\in[0:N-1], u\in \cU^{(Q=q)}$, thereby ensuring $\bar{P}'_{Q^+|Q}(q+1|q)\mspace{-3mu}=\mspace{-3mu}1$. At the last node, however, $\bar{P}'$ mirrors $\bar{P}^*$ by setting $\bar{P}'_{U^+|Q}(0|N)\mspace{-3mu}=\mspace{-3mu}\bar{P}^*_{U^+|Q}(0|N)$. The resulting stationary distribution of the last node, induced by the policy of $\bar{P}'$, is calculated as $\bar{P}'_{Q}(N)\mspace{-3mu}=\mspace{-3mu} (N \bar{P}^*_{U^+|Q}(0|N)+1)^{-1}$.

Finally, taking the limit as $N\to \infty$ in \eqref{eq:gap}, we obtain
\begin{align}
    0 \le \lim_{N\to \infty} b_N-a_N \le \lim_{N\to \infty} \max_{p\in[0,1]} \frac{H(p)}{N p +1} \mspace{-3mu}=\mspace{-3mu}0. 
\end{align}
Consequently, \eqref{eq:conv_part1} holds, thereby concluding the proof.
\end{proof}

\begin{remark}    
In \eqref{eq:gap}, we derive an upper bound on the gap between the upper and lower bounds:
\begin{align}
    b_N-a_N \le \max_{p\in [0,1]}\frac{H(p)}{N p +1} \triangleq \psi(N)=O(N). \label{eq:gapPrecision}
\end{align}
Notably, this upper bound is independent of $\eta$, and guarantees
that the precision $\psi(N)$ for calculating $C_{\text{BEHC}}$ can be achieved without explicitly computing $a_N$ or $b_N$ using convex optimization algorithms. E.g., with $N=10,000$, a precision of $\psi(10,000)=0.0010432 \approx 1e-3$ is guaranteed for all values of $\eta$. Yet, from Table~\ref{table:noiseless_BEHC_full}, we observe that as $\eta$ increases, a smaller $N$ required to achieve a given precision, and vice versa. While the upper bound provided by $\psi(N)$ is generally loose, especially for large values of $\eta$, its independence from $\eta$ ensures that a desired precision can still be achieved by appropriately selecting $N$, even when $\eta$ is very small.
\end{remark}

\begin{remark}
    The convex optimization problems in \eqref{optProb_EH_LB} and \eqref{optProb_EH_UB} involve $O(N^2)$ constraints. While increasing $N$ improves precision per \eqref{eq:gapPrecision}, it also incurs quadratic space complexity, making large $N$ computationally demanding. The main computational burden lies in constructing the constraints, whereas solving the optimization problem itself, once set, is relatively fast using a convex optimization algorithm.
\end{remark}

\section{Achievable Rates for Noisy BEHCs}
\label{sec:Noisy}
In this section, we consider noisy BEHCs.
Using the MDP formulation from \cite[Th.~4]{shemuel2024finite} and the $Q$-graph lower bound in Theorem~\ref{theorem:qgraph_LB}, achievable rates for BEHCs with any DMC and feedback can be evaluated numerically. As an example, we apply the VIA for the MDP to the BSC with $\left|\mathcal{U}\right|\mspace{-3mu}=\mspace{-3mu}2$ and the same $f(u^+,s)$ as in Eq. \eqref{eq:x_behc}. The numerical results, summarized in Table~\ref{table:BSC_BEHC}, show achievable rates as a function of the EH parameter $\eta$ and the channel crossover probability $p$. Notably, for $\eta \mspace{-3mu}=\mspace{-3mu}1$, the achievable rate matches the capacity of the standard BSC (with/ without feedback), given by $1\mspace{-3mu}-\mspace{-3mu}H(p)$.

\begin{table}[t]
\caption{Achievable rates for the BEHC$(\eta)$ over a BSC($p$)}
\label{table:BSC_BEHC}
\begin{center}
\begin{tabular}{|c||c|c|c|c|c|}
\hline 
\multicolumn{1}{|c||}{\textbf{$\bm{\eta} \backslash \mathbf{p}$}} & \textbf{0.1} & \textbf{0.2} & \textbf{0.3} & \textbf{0.4}\\
\hline \hline
\textbf{0.1} & 0.0724 & 0.0331 & 0.0132 & 0.0032\\
\hline
\textbf{0.2} & 0.1437 & 0.0639 & 0.0280 & 0.0068\\
\hline
\textbf{0.3} & 0.2106 & 0.1044 & 0.0433 & 0.0105\\
\hline
\textbf{0.4} & 0.2737 & 0.1397 & 0.0588 & 0.0143\\
\hline
\textbf{0.5} & 0.3271 & 0.1730 & 0.0736 & 0.0180\\
\hline
\textbf{0.6} & 0.3725 & 0.2022 & 0.0872 & 0.0214\\
\hline
\textbf{0.7} & 0.4105 & 0.2261 & 0.0994 & 0.0243\\
\hline
\textbf{0.8} & 0.4486 & 0.2445 & 0.1074 & 0.0268\\
\hline
\textbf{0.9} & 0.4870 & 0.2608 & 0.1136 & 0.0284\\
\hline
\textbf{1} & 0.5310 & 0.2781 & 0.1187 & 0.0290 \\
\hline
\end{tabular}
\end{center}
\end{table}

More generally, the lower bound computation techniques from \cite{shemuel2024finite} extend to EH models with any DMC, arbitrary finite input/output alphabets, any battery size, and any EH process $E_i$. However, establishing computable upper bounds for general EH models remains challenging, as it requires proving a finite cardinality bound for the auxiliary RV $U$, as done for the BEHC.

\section{Conclusion}
\label{sec:conclusions}
This work resolves the open problem of computing the capacity of the BEHC, a challenge stemming from the encoder's intricate constraints, the decoder's lack of state knowledge, and the requirement for infinitely large memory to track consecutive zero inputs in the full history until they are followed by a $'1'$, which resets the tracking. While prior capacity expressions in the literature are multi-letter and uncomputable, we establish that the capacity can be determined to any desired precision via the $Q$-graph method and convex optimization.

Future research may explore coding schemes for the BEHC, refine upper bounding techniques for broader EH models, and extend the approach to large alphabets or continuous-state EH systems.

\appendices

\section{Proof of the Markov Chain 
$X_i-(\hat{U}_{K_{i-1}}^i,Q_{i-1})-(\hat{U}^{K_{i-1}-1},X^{i-1})$}
\label{appendix:lemma_boosted}
\begin{lemma}
\label{lem:boosted_markov}
For the boosted BEHC($N$), the following holds for any time~$i$
\begin{align}
        P(x_i|\hat{u}^i,x^{i-1})=P(x_i|\hat{u}_{k_{i-1}}^i,q_{i-1}(x_{i-N}^{i-1})), \quad \forall \hat{u}^i,x^i,\label{eq:boosted_markov}
\end{align}
where $q_{i-1}$ is a vertex of the $Q$-graph illustrated in Fig.~\ref{fig:EH_UB_Graph}.
\end{lemma}

\begin{proof}
We address each of the two cases separately. For Case $(i)$, i.e., $X_{K_{i-1}}^{i-1}\mspace{-4mu}=\mspace{-4mu}(1,\mathbf{0}^{i-1-{K_{i-1}}})$, we have $S_{K_{i-1}-1}\mspace{-4mu}=\mspace{-4mu}1$ and
\begin{align}
    P(x_i|\hat{u}^i,\mspace{-3mu}x^{i-1})\mspace{-4mu}& \stackrel{(a)}=\mspace{-4mu}P(x_i|\hat{u}^i,x^{i-1},s_{k_{i-1}-1}\mspace{-4mu}=\mspace{-4mu}1,q_{i-1},k_{i-1}) \nn\\
    &\stackrel{(b)}=P(x_i|\hat{u}_{k_{i-1}}^i,x_{k_{i-1}}^{i-1},s_{k_{i-1}-1}\mspace{-4mu}=\mspace{-4mu}1, q_{i-1},k_{i-1}) \nn\\    
    & \stackrel{(c)}=P(x_i|\hat{u}_{k_{i-1}}^i, q_{i-1},k_{i-1}). \label{eq:opt1}
\end{align}
Step (a) follows because $q_{i-1}$ and $k_{i-1}$ are deterministic functions of $x_{i-N}^{i-1}$. In Step (b), the Markov chain follows from an identical derivation of \cite[Lemma~21]{PermuterWeissmanGoldsmith09} that includes $\hat{U}^i$ in the conditioning. Step (c) follows since $x_{k_{i-1}}^{i-1}$ and $s_{{k_{i-1}}-1}$ are deterministic functions of $k_{i-1}$.

For Case $(ii)$, i.e., $x_{i-N}^{i-1}\mspace{-4mu}=\mspace{-4mu}\mathbf{0}^N$, we have $S_{i-1}\mspace{-4mu}=\mspace{-4mu}1$ and
\begin{align}
    &P(x_i|\hat{u}^i,\mspace{-3mu}x^{i-N-1},\mspace{-3mu}x_{i-N}^{i-1}\mspace{-4mu}=\mspace{-4mu}\mathbf{0}^N) \mspace{-4mu}=\mspace{-4mu}P(x_i|\hat{u}_i,\mspace{-3mu}s_{i-1}\mspace{-4mu}=\mspace{-4mu}1,\mspace{-3mu} q_{i-1}\mspace{-4mu}=\mspace{-4mu}N) \nn\\
    &=\mathbbm{1}\{x_i\mspace{-4mu}=\mspace{-4mu}f(\hat{u}_i,\mspace{-3mu}s_{i-1}\mspace{-4mu}=\mspace{-4mu}1)\}
    =P(x_i|\hat{u}_i,\mspace{-3mu}q_{i-1}=N). \label{eq:opt2}
\end{align}
Consequently, from \eqref{eq:opt1} and \eqref{eq:opt2}, we obtain for any $\hat{u}^i,x^i$:
\begin{align}
    P(x_i|\hat{u}^i,\mspace{-3mu}x^{i-1})&=P(x_i|\hat{u}_{k_{i-1}}^i,\mspace{-3mu}q_{i-1}(x_{i-N}^{i-1}),\mspace{-3mu} k_{i-1}(x_{i-N}^{i-1})) \nn\\
    &=P(x_i|\hat{u}_{k_{i-1}}^i,\mspace{-3mu}q_{i-1}(x_{i-N}^{i-1})),
\end{align}
where the last step follows from the one-to-one mapping between $k_{i-1}(x_{i-N}^{i-1})$ and $q_{i-1}(x_{i-N}^{i-1})$.
\end{proof}

\section{Proof of Lemma \ref{lem:tryTransmit1}}
\label{Appendix:prf_lem_tryTransmit1}
Throughout the proof, we denote the special component by 
$\tilde{\hat{u}}_l=b$, where $l\in [k_{i-1}\mspace{-3mu}:\mspace{-3mu} i\mspace{-3mu}-\mspace{-3mu}1]$, i.e., $f_{\tilde{\hat{u}}_l}(s_{l-1}\mspace{-4mu}=\mspace{-4mu}1)\mspace{-4mu}=\mspace{-4mu}1$, and any vector containing it by $\tilde{\hat{u}}_m^j\triangleq (\hat{u}_m,\dots,\hat{u}_{l-1},\tilde{\hat{u}}_l,\hat{u}_{l+1},\dots,\hat{u}_j), m \le l \le j$. We aim to prove that
$P(x_i|\tilde{\hat{u}}^i,x^{i-1})\mspace{-4mu}=\mspace{-4mu}P(x_i|\tilde{\hat{u}}_l^i,q_{i-1})$, assuming $x_{i-N}^{i-1}\mspace{-3mu}\neq \mspace{-3mu}\mathbf{0}^N$. 
\begin{proof}
Assume there exists an input $'1'$ among $x_{i-N}^{i-1}$. Then, 
\begin{align}
    & P(x_i|\tilde{\hat{u}}^i,x^{i-1})=\sum\nolimits_{s_l} P(s_l|\tilde{\hat{u}}^i,x^{i-1})P(x_i|\tilde{\hat{u}}^i,x^{i-1},s_l) \nn\\
    & \stackrel{(a)}= \frac{\sum_{s_l} P(\hat{u}_{l+1}^{i},x_{l+1}^{i-1}|\tilde{\hat{u}}^l,x^l,s_l) P(s_l|\tilde{\hat{u}}^l,x^l) P(x_i|\tilde{\hat{u}}^i,x^{i-1},s_l)}{\sum_{s_l} P(\hat{u}_{l+1}^{i},x_{l+1}^{i-1},s_l|\tilde{\hat{u}}^l,x^l)}  \nn\\
    & = \frac{\sum_{s_l} P(\hat{u}_{l+1}^{i},x_{l+1}^{i-1}|\tilde{\hat{u}}^l,x^l,s_l) P(s_l|\tilde{\hat{u}}^l,x^l) P(x_i|\tilde{\hat{u}}^i,x^{i-1},s_l)}{\sum_{s_l} P(\hat{u}_{l+1}^{i},x_{l+1}^{i-1}|\tilde{\hat{u}}^l,x^l,s_l) P(s_l|\tilde{\hat{u}}^l,x^l)} \nn\\
    & \stackrel{(b)}=\frac{\sum\limits_{s_l}  P(s_l|\tilde{\hat{u}}^l\mspace{-3mu},\mspace{-3mu}x^l)  \mspace{-10mu}\prod\limits_{j=l+1}^{i} \mspace{-10mu}P(\hat{u}_j|\tilde{\hat{u}}^{j-1}\mspace{-3mu},\mspace{-3mu}x^{j-1}) P(x_j|\tilde{\hat{u}}^j\mspace{-3mu},\mspace{-3mu}x^{j-1}\mspace{-3mu},\mspace{-3mu}q_{l}\mspace{-3mu},\mspace{-3mu}s_l) }{\sum\limits_{s_l}  P(s_l|\tilde{\hat{u}}^l\mspace{-3mu},\mspace{-3mu}x^l)  \mspace{-10mu} \prod\limits_{j=l+1}^{i} \mspace{-10mu} P(\hat{u}_j|\tilde{\hat{u}}^{j-1}\mspace{-3mu},\mspace{-3mu}x^{j-1})  \mspace{-10mu} \prod\limits_{j=l+1}^{i-1} \mspace{-10mu} P(x_j|\tilde{\hat{u}}^j\mspace{-3mu},\mspace{-3mu}x^{j-1}\mspace{-3mu},\mspace{-3mu}q_{l}\mspace{-3mu},\mspace{-3mu}s_l) }  \nn\\
    & \stackrel{(c)}= \frac{\sum_{s_l} P(s_l|\tilde{\hat{u}}^l,x^l) \prod_{j=l+1}^{i} P(x_j|\hat{u}_{l+1}^j,x_{l+1}^{j-1},q_{l},s_l)}{\sum_{s_l} P(s_l|\tilde{\hat{u}}^l,x^l) \prod_{j=l+1}^{i-1} P(x_j|\hat{u}_{l+1}^j,x_{l+1}^{j-1},q_l,s_l)}  \nn\\ 
    & \stackrel{(d)}= \frac{\sum_{s_l} P(s_l|q_{l-1},x_l) \mspace{-5mu}\prod_{j=l+1}^{i} \mspace{-5mu} P(x_j|\hat{u}_{l+1}^j,x_{l+1}^{j-1},q_{l},s_l)}{\sum_{s_l} P(s_l|q_{l-1},x_l) \mspace{-10mu} \prod_{j=l+1}^{i-1} \mspace{-5mu}P(x_j|\hat{u}_{l+1}^j,x_{l+1}^{j-1},q_l,s_l)}     
    . \label{eq:yGivenPastEH_last}
\end{align}
Step (a) follows from Bayes' rule and the law of total probability. Step (b) follows from the Markov chain $\hat{U}_j-(\hat{U}^{j-1},X^{j-1})-\tilde{S}^{j-1}$ for any time $j$, implied by \eqref{eq:joint_boosted}, and the fact that $q_{l}$ is a function of $(x_{l-N+1}^{l})$. In Step (c), $\prod_{j=l+1}^{i} P(\hat{u}_j|\tilde{\hat{u}}^{j-1},x^{j-1})$ cancels out, and it follows from the claim that the Markov $P(x_j|\tilde{\hat{u}}^j,x^{j-1},s_l)=P(x_j|\hat{u}_{l+1}^j,x_{l+1}^{j-1},s_l) 
$ holds for any $j>l$. This claim follows directly from an identical derivation of \cite[Lemma~21]{PermuterWeissmanGoldsmith09}, with $\hat{U}^j$ in the conditioning. Step (d) holds by
\begin{align}
& P(s_l\mspace{-3mu}=\mspace{-3mu}0|\tilde{\hat{u}}^l,\mspace{-3mu}x^l)=\mspace{-7mu}\sum\nolimits_{s_{l-1}}\mspace{-17mu} P(s_{l-1}|\tilde{\hat{u}}^l,\mspace{-3mu}x^l) P(s_l\mspace{-3mu}=\mspace{-3mu}0|\tilde{\hat{u}}^l,\mspace{-3mu} x^l, \mspace{-3mu} s_{l-1}, \mspace{-3mu} q_{l-1}) \nn\\
& \stackrel{(*)}=\sum\nolimits_{s_{l-1}} 
P(s_{l-1} | \tilde{\hat{u}}^l, x^l)
\bar{\eta} \quad=\bar{\eta}, \nn 
\end{align}
where $(*)$ follows from \eqref{eq:modifiedChannelLaw} and the assumptions: 
\begin{itemize}
    \item $x_{i-N}^{i-1} \neq \mathbf{0}^N$, which implies $q_{l-1} \neq N$.
    \item $f_{\tilde{\hat{u}}_l}(s_{l-1} = 1) = 1$, which implies $x_l = s_{l-1}$.
\end{itemize}

Since $\hat{u}^{l-1},x^{l-1}$ do not appear in \eqref{eq:yGivenPastEH_last}, it follows that $P(x_i|\tilde{\hat{u}}^i,x^{i-1})=P(x_i|\tilde{\hat{u}}_{l}^i,x_{l+1}^{i-1})$. From Lemma~\ref{lem:boosted_markov} and the assumption that $l\ge k_{i-1}$, we also have:\\
$ P(x_i|\tilde{\hat{u}}^i,x^{i-1})=P(x_i|\tilde{\hat{u}}_{k_{i-1}}^i,q_{i-1}(x^{i-1}_{i-N})).$ Since these two expressions are equal, the proof is complete.
\end{proof}
\bibliographystyle{IEEEtran}
\bibliography{ref}
\clearpage
\setcounter{totalnumber}{1}  

\end{document}